\documentclass[sigconf]{acmart}
\usepackage{color}
\usepackage{algorithm}
\usepackage{algpseudocode}
\usepackage{psfrag}
\usepackage[latin1]{inputenc}
\usepackage{tikz}
\usetikzlibrary{calc}
\usetikzlibrary{arrows,automata}

\usepackage{amsmath}
\usepackage{pifont} 
\usepackage{booktabs}
\usepackage{multirow}
\newtheorem{definition}{Definition}
\newtheorem{proposition}{Proposition}
\newtheorem{theorem}{Theorem}
\newtheorem{lemma}{Lemma}

\newtheorem{example}{Example}

\newcommand{\R}{\mathbb{R}}
\newcommand{\Rn}{\mathbb{R}^n}

\newcommand{\Gr}[1]{{G_I}}

\newcommand\eatpunct[1]{}

\DeclareMathOperator*{\argmax}{argmax} 

\newcommand{\operator}[1]{{\normalfont \texttt{#1}}}
 
\newcommand\eval[1]{
  \hspace{2pt}
  \tikz[baseline=(n.base)]{\node(n)[inner sep=1.5pt]{$#1$};
    \draw[line width=0.25mm](n.west)--(n.south west)--(n.south east)--(n.east);
  }
  \hspace{2pt}
}

\makeatletter
\let\OldStatex\Statex
\renewcommand{\Statex}[1][3]{%
  \setlength\@tempdima{\algorithmicindent}%
  \OldStatex\hskip\dimexpr#1\@tempdima\relax}
\makeatother

\copyrightyear{2020}
\acmYear{2020}
\setcopyright{acmlicensed}
\acmConference[HSCC '20]{23rd ACM International Conference on Hybrid Systems: Computation and Control}{April 22--24, 2020}{Sydney, NSW, Australia}
\acmBooktitle{23rd ACM International Conference on Hybrid Systems: Computation and Control (HSCC '20), April 22--24, 2020, Sydney, NSW, Australia}
\acmPrice{15.00}
\acmDOI{10.1145/3365365.3382192}
\acmISBN{978-1-4503-7018-9/20/04}

\begin{document}

\title{Utilizing Dependencies to Obtain Subsets \\of Reachable Sets}

\author{Niklas Kochdumper}
\email{niklas.kochdumper@tum.de}
\affiliation{%
  \institution{Technische Universit\"at M\"unchen}
}

\author{Bastian Sch\"urmann}
\email{bastian.schuermann@tum.de}
\affiliation{%
  \institution{Technische Universit\"at M\"unchen}
}

\author{Matthias Althoff}
\email{althoff@tum.de}
\affiliation{%
  \institution{Technische Universit\"at M\"unchen}
}

\begin{abstract}
	Reachability analysis, in general, is a fundamental method that supports
formally-correct synthesis, robust model predictive control, set-based
observers, fault detection, invariant computation, and conformance
checking, to name but a few. In many of these applications, one
requires to compute a reachable set starting within a previously
computed reachable set. While it was previously required to re-compute
the entire reachable set, we demonstrate that one can leverage the
dependencies of states within the previously computed set. As a result,
we almost instantly obtain an over-approximative subset of a previously computed reachable set by evaluating analytical maps. The advantages of our novel method are demonstrated for falsification of systems, optimization over reachable sets, and synthesizing safe maneuver automata. In all of these
applications, the computation time is reduced significantly.
\end{abstract}



\maketitle


\section{Introduction} 
\label{sec:Introduction}

In this paper, we present a novel method to directly extract a reachable set within a pre-computed one as depicted in Fig.~\ref{fig:VisBasicApproach}. To achieve this, we conservatively abstract the original dynamics by a polynomial right-hand side and represent sets by polynomial zonotopes \cite{Althoff2013a}. Since polynomial zonotopes preserve the relation between the reachable states and the states in the initial set, we can extract the reachable set $\widehat{\mathcal{R}}$ for any initial set $\widehat{\mathcal{X}}_0 \subseteq \mathcal{X}_0$ directly from the reachable set $\mathcal{R}$ by evaluating an analytical equation (see Fig.~\ref{fig:VisBasicApproach}), which is computationally more efficient than computing $\widehat{\mathcal{R}}$ with a reachability algorithm. As we demonstrate in Sec.~\ref{sec:Applications}, this method offers great advantages for applications where reachable sets have to be computed for many different subsets $\widehat{\mathcal{X}}_0 \subseteq \mathcal{X}_0$, like \textit{safety falsification}, \textit{optimization over reachable sets}, and \textit{motion-primitive based control}. 

\begin{figure}
\begin{center}
	\includegraphics[width = 0.35 \textwidth]{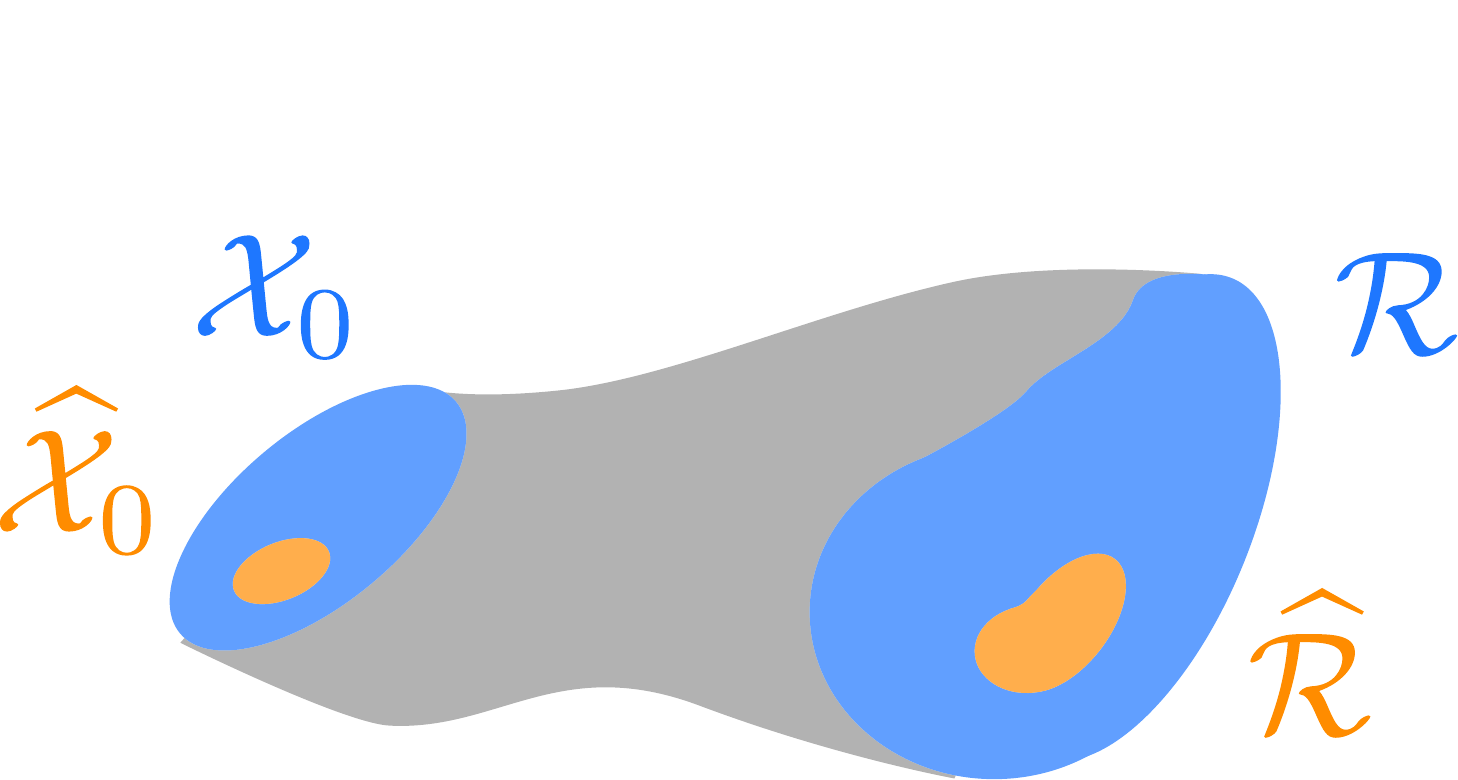}
	\caption{Given a reachable set $\mathcal{R}$ for a set of initial states $\mathcal{X}_0$ and a subset of initial states $\widehat{\mathcal{X}}_0$, we can obtain $\widehat{\mathcal{R}}$ without any reachability analysis.}
	\label{fig:VisBasicApproach}
	\end{center}
\end{figure}

\subsection{State of the Art}

Reachability algorithms for linear systems and hybrid systems with linear continuous dynamics are mostly based on the propagation of reachable sets. These algorithms use a large variety of convex set representations, like polytopes \cite{Frehse2011}, zonotopes \cite{Girard2006}, ellipsoids \cite{Kurzhanski2000}, support functions \cite{Girard2008b}, and star sets \cite{Duggirala2016}. Further, other approaches use simulations to compute reachable set \cite{Bak2017b,Duggirala2016}. The examples of tools for reachability analysis of linear systems are C2E2 \cite{Duggirala2015}, CORA \cite{Althoff2015a}, HyDRA \cite{Schupp2017}, Hylaa \cite{Bak2017}, Julia Reach \cite{Bogomolov2019a}, SpaceEx \cite{Frehse2011}, and XSpeed \cite{Ray2015}. 

The reachability algorithms for nonlinear systems can be categorized into four groups: invariant generation, optimization-based approaches, abstraction in solution space, and abstraction in state space. Since any invariant set which includes the initial set is also a reachable set, approaches for invariant generation can be used for reachability analysis \cite{Kong2017,Matringe2010,Liu2011}. The optimization-based approaches reformulate reachability analysis as an optimization problem \cite{Chutinan2003,Mitchell2005}. Thus, the approach in \cite{Chutinan2003} optimizes the outward translation of polytope halfspaces to obtain a flowpipe, whereas \cite{Mitchell2005} expresses reachability analysis with Hamilton-Jacobi equations. Other approaches abstract the solution space directly: The work in \cite{Duggirala2016} uses validated simulations for the construction of bounded flowpipes, and \cite{Prabhakar2011} approximates the solution of the ODE with Bernstein polynomials; Taylor models computed from iterations such as the Picard iteration were initially proposed in \cite{Hoefkens2001,Makino2009}, and later extended to include uncertain inputs \cite{Chen2012}. Approaches based on an abstraction of the state space compute simplified differential equations to which a compensating uncertainty is added. Often, nonlinear ODEs are abstracted by a hybrid automaton with linear dynamics \cite{Asarin2007}. Other methods linearize the nonlinear dynamics on-the-fly \cite{Dang2010,Althoff2008c}. Recent approaches extend this concept to the abstraction of the nonlinear dynamics by polynomials \cite{Althoff2013a}, which results in a tighter enclosure of the reachable set. Examples of tools for reachability analysis of nonlinear systems are Adriadne \cite{Benvenuti2014}, C2E2 \cite{Duggirala2015}, CORA \cite{Althoff2015a}, DynIbex \cite{Sandretto2016}, Flow* \cite{Chen2013}, Isabelle/HOL \cite{Immler2015}, and Julia Reach \cite{Bogomolov2019a}.

We show in this work how the relation between the initial and reachable states can be preserved during reachability analysis. For \textit{abstract interpretation}, the work in \cite{Goubault2015} demonstrates how the parameterization of zonotopes can be utilized to preserve relations between inputs and outputs of computer programs. In \cite{Chen2014}, the relation between the initial and reachable states is used to compute inner-approximations of reachable sets. Similarly, the approach in \cite{Chen2012} utilizes this relation to tighly enclose the intersections with guard sets for hybrid system reachability analysis. 

Our new method is not only applicable to reachability analysis but also falsification. If a computed reachable set violates a specification, falsification aims to provide the initial states and the input signals that lead to the violation. The problem of finding such initial states and input signals is known as \textit{safety falsification}. For linear systems with piecewise constant inputs, \cite{Bak2017b} extracts falsifying initial states and input signals from the computed reachable set by solving a linear program. A similar concept is used in \cite{Bogomolov2019}, where reachability analysis is utilized to efficiently determine falsifying trajectories for hybrid systems with linear continuous dynamics. The set of initial states resulting in a falsification can also be computed with backward reachability analysis using inner-approximations of reachable sets. Approaches to compute inner-approximations of reachable sets exist for linear systems \cite{Girard2006} as well as for nonlinear systems \cite{Goubault2019,Li2018a,Chen2014}. 

In this paper, we preserve the relation between the initial and the reachable states for nonlinear systems. We use polynomial zonotopes instead of Taylor models since the former is a generalization of Taylor models \cite[Prop.~4]{Kochdumper2019}. Our novel method allows us to obtain subsets within a pre-computed reachable set almost instantly. Using numerical examples, we demonstrate that our method reduces the computation time for \textit{safety falsification}, \textit{optimization over reachable sets}, and \textit{motion-primitive based control} significantly. To some extent this work proposes a novel method for unifying reachability analysis and falsification.

\subsection{Notation}
\label{subsec:notation}

Sets are denoted by calligraphic letters, matrices by uppercase letters, vectors by lowercase letters, and set operations by typewriter font (e.g., \operator{interval}). Given a vector $b \in \mathbb{R}^n$, $b_{(i)}$ refers to the $i$-th entry. Given a matrix $A \in \mathbb{R}^{n \times m}$, $A_{(i,\cdot)}$ represents the $i$-th matrix row, $A_{(\cdot,j)}$ the $j$-th column, and $A_{(i,j)}$ the $j$-th entry of matrix row $i$. The concatenation of two matrices $C$ and $D$ is denoted by $[C~D]$. The symbols $\mathbf{0}$ and $\mathbf{1}$ represent matrices of zeros and ones of proper dimension and the empty matrix is denoted by $[~]$. Left multiplication of a matrix $M \in \mathbb{R}^{m \times n}$ with a set $\mathcal{S} \subset \mathbb{R}^n$ is defined as $M \otimes \mathcal{S} = \big \{ M s ~ \big | ~ s \in \mathcal{S} \big \}$, the Minkowski addition of two sets $\mathcal{S}_1 \subset \mathbb{R}^n$ and $\mathcal{S}_2 \subset \mathbb{R}^n$ is defined as $\mathcal{S}_1 \oplus \mathcal{S}_2 = \big \{ s_1 + s_2 ~\big |~ s_1 \in \mathcal{S}_1, s_2 \in \mathcal{S}_2 \big\}$, and the Cartesian product of two sets $\mathcal{S}_1 \subset \mathbb{R}^n$ and $\mathcal{S}_2 \subset \mathbb{R}^m$ is defined as $\mathcal{S}_1 \times \mathcal{S}_2 = \big \{ [s_1~s_2]^T ~\big|~ s_1 \in \mathcal{S}_1, s_2 \in \mathcal{S}_2 \big\}$. Given two set operations \operator{A} and \operator{B}, and a set $\mathcal{S} \subset \Rn$, the composition of the set operations is denoted by $\operator{A}(\operator{B}(\mathcal{S})) = (\operator{A} \circ \operator{B})(\mathcal{S})$. The power set of a set $\mathcal{S} \subset \Rn$ is denoted by $2^{\mathcal{S}}$. We further introduce a n-dimensional interval as $\mathcal{I} := [l,u],~ \forall i ~ l_{(i)} \leq u_{(i)},~ l,u \in \mathbb{R}^n$. The unit hypercube $[-\mathbf{1},\mathbf{1}] \subset \R^p$ is denoted by $\mathcal{I}_p$. Given a center vector $c \in \Rn$ and a generator matrix $G \in \R^{n \times m}$, a zonotope is $\mathcal{Z} := \big \{ c + \sum_{i=1}^m \alpha_i ~G_{(\cdot,i)} ~\big |~ \alpha_i \in [-1,1] \big \}$. For a concise notation, we use the shorthand $\mathcal{Z} = \langle c,G \rangle_Z$.


\section{Set Representation} 
\label{sec:SetRepresentation}

\subsection{Parameterization}

\begin{table*}
\begin{center}
\caption{Characterization of set representations w.r.t. parameterization and dependency preservation of set operations.}
\label{tab:parameterization}
\begin{tabular}{ l c c c c c c}
 \toprule
 \multirow{2}{3.5cm}{\textbf{Set Representation}} & \textbf{Parameter-} & \multicolumn{4}{c}{\textbf{Dependency Preservation}} & \\ \cmidrule{3-7}
		& \textbf{ization} & \textbf{Linear Trans.} & \textbf{Minkowski Sum} & \textbf{Cart. Product} & \textbf{Quad. Map} \\ \midrule 
 Interval & $\times$ &  &  &  &  \\
 Zonotopes & $\surd$ & $\surd$ & $\surd$ & $\surd$ & $\surd$ \\ 
 Polytopes (H-Rep.) & $\times$ &  &  &  &  \\
 Polytopes (V-Rep.) & $\surd$ & $\surd$ & $\times$ & $\times$ & $\times$ \\
 Ellipsoids & $\surd$ & $\surd$ & $\times$ & $\times$ & $\times$ \\
 Support Functions & $\times$ &  &  &  &  \\
 Level Sets & $\times$ & & & & \\
 Star Sets & $\surd$ & $\surd$ & $\surd$ & $\surd$ & $\times$ \\
 Taylor Models & $\surd$ & $\surd$ & $\surd$ & $\surd$ & $\surd$ \\
 Polynomial Zonotopes & $\surd$ & $\surd$ & $\surd$ & $\surd$ & $\surd$ \\ 
 \bottomrule 
\end{tabular}
\end{center}
\end{table*}

A prerequisite for preserving relations between states is that the states inside the initial set are parameterized. As shown in Table~\ref{tab:parameterization}, not all set representations fulfill this requirement. We demonstrate this exemplary for the vertex (V-representation) and the halfspace-representation (H-representation) of a polytope; next are the definitions.

\begin{definition}
	(V-Representation) Given $m$ vertices $v_i \in \mathbb{R}^n$, the vertex representation of $\mathcal{P} \subset \mathbb{R}^n$ is defined as
	\begin{equation*}
		 \mathcal{P} := \bigg\{  \sum_{i=1}^{m} \delta_i v_i ~ \bigg |~ \delta_i \geq 0, ~ \sum_{i=1}^m \delta_i = 1 \bigg\}.
	\end{equation*}
	\label{def:Vrep}
\end{definition}

We use the shorthand $\mathcal{P} = \langle [v_1~\dots~v_m] \rangle_V$.

\begin{definition}
	(H-Representation) Given a matrix $A \in \mathbb{R}^{m \times n}$ and a vector $b \in \mathbb{R}^m$, the halfspace representation of $\mathcal{P} \subset \mathbb{R}^n$ is defined as 
	\begin{equation*}
		 \mathcal{P} := \big \{ x ~\big |~ A x \leq b \big \}.
	\end{equation*}
\end{definition}

Each point $p \in \mathcal{P}$ can be parameterized by specific values $\overline{\delta}_i$ when using the V-representation (not possible for the H-representation):
\begin{equation}
	p =  \sum_{i=1}^{m} \overline{\delta}_i v_i, ~~ \overline{\delta}_i \geq 0, ~ \sum_{i=1}^m \overline{\delta}_i = 1.
	\label{eq:paramVrep}
\end{equation}
In general, the above parameterization is not unique \cite{Schuermann2016a}. For parameterized sets, we introduce evaluation functions.

\begin{definition}
	(Evaluation Function) Given a set $\mathcal{S} \subset \Rn$ that is parameterized by the parameter vector $d \in \mathcal{D} \subset \R^m$, the evaluation function $\eval{\mathcal{S}}: \mathcal{D} \to 2^\mathcal{S}$ returns the set $\overline{\mathcal{S}}$ that corresponds to a specific value $\overline{d} \in \mathcal{D}$ of the parameter vector $d$:
	\begin{equation*}
		\eval{\mathcal{S}}(\overline{d}) = \overline{\mathcal{S}},
	\end{equation*}
	where the parameter domain $\mathcal{D} \subset \R^m$ satisfies 
	\begin{equation*}
		\bigcup_{d \in \mathcal{D}} \eval{\mathcal{S}}(d) = \mathcal{S}.
	\end{equation*}
\end{definition}

\begin{example}
For a polytope $\mathcal{P} = \langle [v_1~\dots~v_m] \rangle_V$ in V-representation, the parameter domain is
\begin{equation*}
	\mathcal{D} = \bigg\{ [\delta_1~\dots~\delta_m]^T ~ \bigg| ~ \delta_i \geq 0, ~ \sum_{i=1}^m \delta_i = 1 \bigg\}
\end{equation*}
and the evaluation function is
\begin{equation}
	\eval{\mathcal{P}}(\delta) = \bigg \{ \sum_{i=1}^{m} \delta_i v_i \bigg \}.
	\label{eq:evalFunPoly}
\end{equation}
\end{example}

\subsection{Dependency Preservation}

We require a set representation that preserves the relation between states for all relevant set operations. First, we demonstrate dependency preservation for linear maps of V-representations:
\begin{example}
	Given a scalar $M \in \R$ and a one-dimensional polytope $\mathcal{P} = \langle [v_1,v_2] \rangle_V$, its linear map is computed as 
	\begin{equation}
			M \otimes \mathcal{P} = \langle [M v_1~M v_2] \rangle_V.
			\label{eq:linMapPoly}
	\end{equation}
	Let us introduce the polytope $\mathcal{P} = \langle [-1~3] \rangle_V$, the point $p = 2 \in \mathcal{P}$, and the scalar $M = 2$. According to \eqref{eq:paramVrep}, the point $p \in \mathcal{P}$ can be parameterized by the values $\overline{\delta} = [\overline{\delta_1}~\overline{\delta_2}]^T = [0.25~0.75]^T$. The computation of the linear transformation according to \eqref{eq:linMapPoly} yields
	\begin{equation*}
		M \otimes \mathcal{P} = \langle [\widehat{v}_1~\widehat{v}_{2}] \rangle_V = \langle [-2~6] \rangle_{V}.
	\end{equation*}
	If we evaluate the result for $\overline{\delta}$ corresponding to the point $p$, we obtain
	\begin{equation*}
		\eval{M \otimes \mathcal{P}}(\overline{\delta}) \overset{\eqref{eq:evalFunPoly}}{=} \sum_{i=1}^2 \overline{\delta}_i \widehat{v}_i = 4 = Mp.
	\end{equation*}
	The implementation of the linear map in \eqref{eq:linMapPoly} is therefore dependency-preserving.
\end{example}

Next, we consider a quadratic map, which is not dependency-preserving for the V-representation:

\begin{example}
	Given a scalar $Q \in \R$ and a one-dimensional polytope $\mathcal{P} = \langle [v_1~v_2] \rangle_V$, its quadratic map is computed as
	\begin{equation}
		\begin{split}
			& \operator{sq}(Q,\mathcal{P}) = \big \{ x^T Q x ~\big|~ x \in \mathcal{P} \big\} \\
			& ~ = \begin{cases} \left[ 0~~\max(|v_1|,|v_2|)^2 Q\right], & 0 \in \mathcal{P} \\ \left[ \min(v_1^2,v_2^2) Q~~ \max(v_1^2,v_2^2) Q \right], & \mathrm{otherwise} \end{cases}.
		\end{split}
		\label{eq:quadMapPoly}
	\end{equation} 
	Let us introduce the polytope $\mathcal{P} = \langle [-1~3] \rangle_V$, the point $p = 2 \in \mathcal{P}$, and the scalar $Q = 2$. According to \eqref{eq:paramVrep}, the point $p \in \mathcal{P}$ can be parameterized by the values $\overline{\delta} = [\overline{\delta_1}~\overline{\delta_2}]^T = [0.25~0.75]^T$. Computation of the quadratic map according to \eqref{eq:quadMapPoly} yields
	\begin{equation*}
		sq(Q,\mathcal{P}) = \langle [\widehat{v}_1~\widehat{v}_{2}] \rangle_V  = \langle [0~18] \rangle_V.
	\end{equation*} 
	If we evaluate the computed quadratic map for $\overline{\delta}$ corresponding to the point $p$, we obtain
	\begin{equation*}
		\eval{\operator{sq}(Q,\mathcal{P})}(\overline{\delta}) \overset{\eqref{eq:evalFunPoly}}{=} \sum_{i=1}^2 \overline{\delta}_i \widehat{v}_i = 13.5 \neq \operator{sq}(Q,p) = 8,
	\end{equation*}
	which is not identical to $\operator{sq}(Q,p)$. The above implementation is therefore not dependency-preserving.
\end{example}

Let us finally define dependency preservation:
\begin{definition}
	(Dependency Preservation) Given an implementation of a set operation \operator{A} and a set $\mathcal{S} \subset \Rn$ parameterized by $d \in \mathcal{D} \subset \R^m$, we call the implementation of \operator{A} dependency-preserving if
	\begin{equation}
			\forall d \in \mathcal{D}: ~~ \operator{A} \left( \eval{\mathcal{S}}(d) \right) \subseteq \eval{\operator{A}(\mathcal{S})}(d).
		\label{eq:depPreserving}
	\end{equation}
	\label{def:depPreserving}
\end{definition}

Table~\ref{tab:parameterization} shows a summary of the set operations that are dependency-preserving for different set representations.

\subsection{Set Operation Properties}

We introduce some additional properties for set operations which we require for later derivations. Given two sets $\mathcal{S}_1,\mathcal{S}_2 \subset \R^n$ with $\mathcal{S}_1 \subseteq \mathcal{S}_2$, all unary set operations \operator{A} used in this work satisfy 
\begin{equation}
  \operator{A}(\mathcal{S}_1) \subseteq \operator{A}(\mathcal{S}_2).
  \label{eq:unary}
\end{equation}
Furthermore, given sets $\mathcal{S}_1,\mathcal{S}_2 \subset \R^n$ and $\mathcal{S}_3,\mathcal{S}_4 \subset \R^m$ with $\mathcal{S}_1 \subseteq \mathcal{S}_2$ and $\mathcal{S}_3 \subseteq \mathcal{S}_4$, all binary set operations \operator{B} used in this work satisfy 
\begin{equation}
  \operator{B}(\mathcal{S}_1,\mathcal{S}_3) \subseteq \operator{B}(\mathcal{S}_2,\mathcal{S}_4).
  \label{eq:binary}
\end{equation}
The properties \eqref{eq:unary} and \eqref{eq:binary} equivalently hold for the composition of set operations:
\begin{lemma}
	Given two set operations \operator{A} and \operator{B} that satisfy \eqref{eq:unary}, the composition $\operator{A} \circ \operator{B}$ also satisfies \eqref{eq:unary}.
	\label{lemma:conseqSubset}
\end{lemma}

\begin{proof}
	Since \operator{A} and \operator{B} satisfy \eqref{eq:unary}, it holds for two sets $\mathcal{S}_1,\mathcal{S}_2 \subset \R^n$ with $\mathcal{S}_1 \subseteq \mathcal{S}_2$ that
	\begin{equation*}
		(\operator{A} \circ \operator{B})(\mathcal{S}_1) = \operator{A}(\underbrace{\operator{B}(\mathcal{S}_1)}_{\overset{\eqref{eq:unary}}{\subseteq} \operator{B}(\mathcal{S}_2)}) \overset{\eqref{eq:unary}}{\subseteq} \operator{A}(\operator{B}(\mathcal{S}_2)) = (\operator{A} \circ \operator{B})(\mathcal{S}_2).
	\end{equation*}
	\hfill ~
\end{proof}
The result of Lemma~\ref{lemma:conseqSubset} equally holds for compositions involving binary set operations that satisfy \eqref{eq:binary}. Next, we show that the composition of two dependency-preserving set operations is dependency-preserving:
\begin{lemma}
	The composition $\operator{A} \circ \operator{B}$ of two dependency-\linebreak[4]preserving operations \operator{A} and \operator{B} is dependency-preserving as well.
	\begin{equation}
			\forall d \in \mathcal{D}: ~~ (\operator{A} \circ \operator{B}) \left( \eval{\mathcal{S}}(d) \right) \subseteq \eval{(\operator{A}\circ \operator{B})(\mathcal{S})}(d).
			\label{eq:lemmaComp}
	\end{equation}
	\label{lemma:conseqExecution}
\end{lemma}

\begin{proof}
	Since \operator{B} is dependency-preserving, it holds according to \eqref{eq:depPreserving} that 
	\begin{equation}
			\forall d \in \mathcal{D}: ~~ \operator{B} \left( \eval{\mathcal{S}}(d) \right) \subseteq \eval{\operator{B}(\mathcal{S})}(d).
			\label{eq:lemmaComp1}
	\end{equation}
	Then, using \eqref{eq:lemmaComp1} and \eqref{eq:unary}, it holds that 
	\begin{equation}
		\forall d \in \mathcal{D}: ~~ \operator{A} \left( \operator{B} \left( \eval{\mathcal{S}}(d) \right) \right) \subseteq \operator{A} ( \eval{\operator{B}(\mathcal{S})}(d) ).
		\label{eq:lemmaComp2}
	\end{equation}
	Since \operator{A} is dependency-preserving, it holds according to \eqref{eq:depPreserving} that 
	\begin{equation*}
	\begin{split}
		\forall d \in \mathcal{D}: ~~ & (\operator{A} \circ \operator{B}) ( \eval{\mathcal{S}}(d)) = \operator{A} \left( \operator{B} \left( \eval{\mathcal{S}}(d) \right) \right) \overset{\eqref{eq:lemmaComp2}}{\subseteq} \operator{A} \big( \eval{ \operator{B}(\mathcal{S})}(d) \big) \\
		& \overset{\eqref{eq:depPreserving}}{\subseteq} \eval{\operator{A}(\operator{B}(\mathcal{S}))}(d) = \eval{ (\operator{A} \circ \operator{B})(\mathcal{S})}(d),
	\end{split}	
	\end{equation*}
	which is identical to \eqref{eq:lemmaComp}. \hfill ~
\end{proof}

Since the reachability algorithm used in this work applies dep-endency-preserving set operations only, is follows that the algorithm is dependency-preserving too, as discussed in Sec.~\ref{sec:ReachabilityAnalysis}.

\subsection{Polynomial Zonotopes}

The concept presented in this work requires a parameterized set representation for which all operations used by the reachability algorithm are dependency-preserving. As shown in Table~\ref{tab:parameterization}, a non-convex set representation that satisfies these requirements are polynomial zonotopes. In this work we use their sparse representation \cite{Kochdumper2019}: 
\begin{definition}
(Polynomial Zonotope) Given a generator matrix of dependent generators $G \in \mathbb{R}^{n \times h}$, a generator matrix of independent generators $G_I \in \mathbb{R}^{n \times q}$, and an exponent matrix $E \in \mathbb{Z}_{\ge 0}^{m \times h}$, a polynomial zonotope is defined as  
  \begin{equation}
  	\begin{split}
    \mathcal{PZ} := \bigg\{ & \underbrace{ \sum _{i=1}^h \bigg( \prod _{k=1}^m \alpha _k ^{E_{(k,i)}} \bigg) G_{(\cdot,i)} }_{f_{G,E}(\alpha)} + \sum _{j=1}^{q} \beta _j G_{I(\cdot,j)}~ \bigg| \\ & \alpha_k, \beta_j \in [-1,1] \bigg\}, 
    \end{split}
  \label{eq:polyZonotope}
  \end{equation}
  where $f_{G,E}: \R^m \to \Rn$ is a multivariate polynomial function, and $\alpha = [\alpha_1~\dots~\alpha_m]^T \in \R^m$.
  \label{def:polynomialZonotope}
\end{definition}

The scalars $\alpha_k$ are called dependent factors and the scalars $\beta_j$ independent factors. Consequently, the term $f_{G,E}(\alpha)$ is called the dependent part, and the term $\sum _{j=1}^{q} \beta _j G_{I(\cdot,j)}$ is called the independent part. We introduce the shorthand $\mathcal{PZ} = \langle G,G_I,E \rangle_{PZ}$. 

Every polynomial zonotope can equivalently be represented by a polynomial zonotope without independent generators:
\begin{proposition}
		Given a polynomial zonotope $\mathcal{PZ} = \langle G,G_I,E \rangle_{PZ}$, $\mathcal{PZ}$ can be equivalently represented without independent generators as follows:
		\begin{equation}
			\mathcal{PZ} = \langle G,G_I,E \rangle_{PZ} = \bigg \langle [G~G_I],[~], \begin{bmatrix} E & \mathbf{0} \\ \mathbf{0} & I_q \end{bmatrix} \bigg \rangle_{PZ},
			\label{eq:remIndepGen}
		\end{equation}		
		where $I_q \in \R^{q \times q}$ is the identity matrix.
		\label{prop:remIndepGen}
\end{proposition} 

\begin{proof}
	The result in \eqref{eq:remIndepGen} follows directly from the substitution of the independent factors $\beta_j$ in the definition of polynomial zonotopes in \eqref{eq:polyZonotope} with additional dependent factors $\alpha_{m+1} = \beta_1,\dots,\alpha_{p+q} = \beta_q$. \hfill ~
\end{proof}

Despite the result of Prop.~\ref{prop:remIndepGen}, the independent part of the polynomial zonotope is crucial for order reduction \cite{Kochdumper2019}. For polynomial zonotopes, the points inside the set are parameterized by both the dependent factors $\alpha_k$ and the independent factors $\beta_j$. Using Prop.~\ref{prop:remIndepGen}, we assume without loss of generality that the initial set has no independent generators. Consequently, it is sufficient to choose the parameter vector as $d = \alpha = [\alpha_1~\dots~\alpha_m]^T$, the parameter domain as $\mathcal{D} = \mathcal{I}_m$, and the evaluation function as
\begin{equation}
	\begin{split}
	\eval{\mathcal{PZ}}(\alpha) &= f_{G,E}(\alpha) \oplus \bigg\{ \sum _{j=1}^{q} \beta _j G_{I(\cdot,j)} ~\bigg|~  \beta_j \in [-1,1] \bigg\} \\
	&= \langle f_{G,E}(\alpha), G_I \rangle_Z.
	\end{split}
	\label{eq:evalFunPolyZono}
\end{equation} 

However, finding a parameterization for a point inside a polynomial zonotope can be computational expensive. For reachability analysis, the initial set is often an axis-aligned box, for which the parameterization is unique and trivial to compute as shown by the following example:
\begin{example}
	We consider the polynomial zonotope
	\begin{equation*}
		\mathcal{PZ} = \bigg \{ \underbrace{ \begin{bmatrix} 1 \\ 0 \end{bmatrix} \alpha_1 + \begin{bmatrix} 0 \\ 1 \end{bmatrix} \alpha_2 }_{f_{G,E}(\alpha)}~ \bigg | ~ \alpha_1,\alpha_2 \in [-1,1] \bigg \}
	\end{equation*}
	and the point $p = [0.5~0.4]^T$. Trivially, the point $p$ can be parameterized with $\overline{\alpha} = [0.5~0.4]^T$, so that $p = f_{G,E}(\overline{\alpha})$.
	\label{ex:polyZonotope}
\end{example}

Next, we integrate parameterized and dependency-preserving set representations in reachability analysis to obtain reachable subsets.

\begin{figure}[h]
\begin{center}
	\includegraphics[width = 0.45 \textwidth]{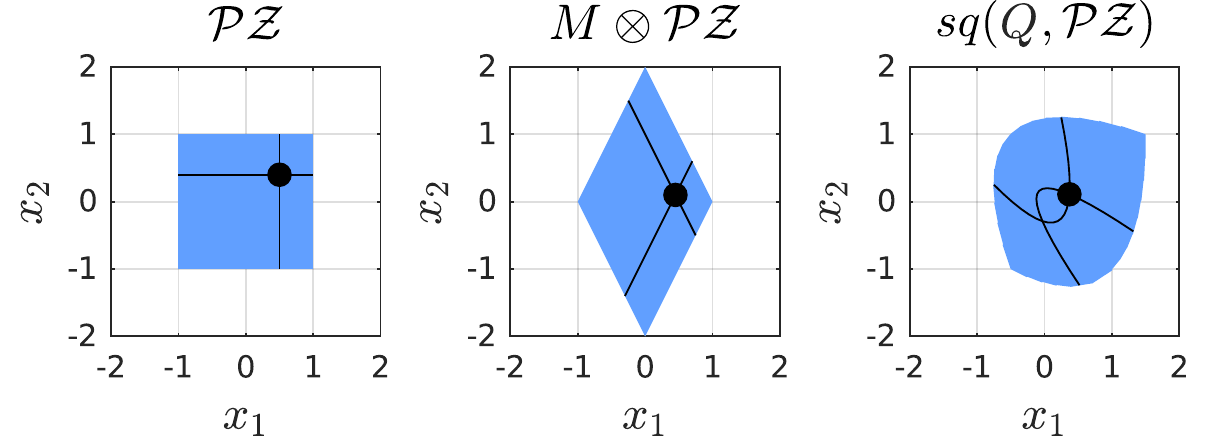}
	\caption{Sets resulting from the application of different set operations to the polynomial zonotope $\mathbf{\mathcal{PZ}}$ from Example~\ref{ex:polyZonotope}. The point $\mathbf{p}$ and its transformations are depicted by the black dots. Black lines correspond to the points where $\mathbf{\alpha_1 = 0.5 = const.}$ and $\mathbf{\alpha_2 = 0.4 = const.}$~.}
	\label{fig:exampleSetOps}
	\end{center}
\end{figure}


\section{Reachability Analysis}
\label{sec:ReachabilityAnalysis}

We first recall some preliminaries for reachability analysis, followed by our novel algorithm for computing reachable subsets.

\subsection{Preliminaries}

In this paper we consider nonlinear systems of the form
\begin{equation}
	\dot x (t) = f(x(t),u(t)), ~ x(t) \in \mathbb{R}^n,~ u(t) \in \mathbb{R}^m,
	\label{eq:system}
\end{equation}
where $x$ is the state vector and $u$ is the input vector. The reachable set is defined as follows:
\begin{definition}
	(Reachable Set) Let $\xi(t,x_0,u(\cdot))$ denote the solution to \eqref{eq:system} for an initial state $x(0) = x_0$ and the input trajectory $u(\cdot)$. The reachable set for an initial set $\mathcal{X}_0 \subset \mathbb{R}^n$ and a set of possible input values $\mathcal{U} \subset \mathbb{R}^m$ is 
	\begin{equation*}
		\mathcal{R}_{\mathcal{X}_0}^e(t) := \big\{ \xi(t,x_0,u(\cdot)) ~\big |~ x_0 \in \mathcal{X}_0, \forall \tau \in [0,t]~ u(\tau) \in \mathcal{U} \big \}.
	\end{equation*} 
\end{definition}
The superscript $e$ on $\mathcal{R}_{\mathcal{X}_0}^e(t)$ denotes the exact reachable set, which cannot be computed for general nonlinear systems. Therefore, we compute a tight over-approximation $\mathcal{R}(t) \supseteq \mathcal{R}_{\mathcal{X}_0}^e(t)$ with the operation \operator{reach} defined by Alg.~\ref{alg:reach}, which is taken from \cite{Kochdumper2019}. Alg.~\ref{alg:reach} is based on the abstraction of the nonlinear function $f(\cdot)$ by a Taylor expansion of order $\kappa$:
\begin{equation*}
	\begin{split}
		\dot x_{(i)} & = f_{(i)}(z(t)) \\
		& \in \left. \sum_{j=0}^{\kappa} \frac{\left((z(t)-z^*)^T \nabla \right)^j f_{(i)}(\tilde z)}{j!} \right|_{\tilde{z} = z^*} \oplus \mathcal{L}_{(i)}(t),
	\end{split}
\end{equation*} 
where  we introduce the shorthand $z = [x^T~u^T]^T$ and the Nabla operator $\nabla = \sum_{i=1}^{n+m} e_i \frac{\partial}{\partial z_{(i)}}$, with $e_i$ being orthogonal unit vectors. The set $\mathcal{L}_{(i)}(t)$ is the Lagrange remainder, defined in \cite[Eq. (2)]{Althoff2013a}, and the vector $z^* \in \mathbb{R}^{n+m}$ is the expansion point for the Taylor series. 

Within Alg.~\ref{alg:reach}, Alg.~\ref{alg:linError} is executed in line~\ref{line:linErrorAlg} to compute the set of abstraction errors. Alg.~\ref{alg:reach} and Alg.~\ref{alg:linError} require the following set operations: \operator{zono} returns an enclosing zonotope, $\boxplus$ denotes the exact addition as defined in \cite[Prop. 9]{Kochdumper2019}, \operator{reduce} denotes the order reduction of a polynomial zonotope (see \cite[Prop. 16]{Kochdumper2019}), and \operator{enlarge} enlarges a set by a given scalar factor $\lambda$. The operation \operator{taylor} returns the matrices $w,A,B,D,E$ storing the coefficients of the Taylor series evaluation of the nonlinear function $f(\cdot)$ at the expansion point $z^*$. The definitions of the operations $\mathcal{R}^{p,\Delta}$ \cite[Eq.~(9)]{Althoff2013a}, $\operator{post}^{\Delta}$ \cite[Sec. 4.1]{Althoff2013a}, \operator{varInputs} \cite[Sec. 4.2]{Althoff2013a}, and \operator{lagrangeRemainder} \cite[Sec. 4.1]{Althoff2013a} are identical to the ones in \cite{Althoff2013a}. The implementations of all required set operations for polynomial zonotopes are dependency-preserving (see Def.~\ref{def:depPreserving}), as shown in Table 1. Exact addition and order reduction are not included in Table~\ref{tab:parameterization} since they do not apply to all set representations. A visualization of dependency preservation is shown in Fig.~\ref{fig:exampleSetOps} for some numerical examples.

\begin{algorithm}[h!tb]
	\caption{$\mathcal{R}(t_f) = \operator{reach}(\mathcal{X}_0,\Psi(\tau_s),z_s^*)$} \label{alg:reach}
	{\raggedright \textbf{Require:} Initial set $\mathcal{X}_0$, sets of initial abstraction errors $\Psi(\tau_s)$, expansion points $z_s^*$, time horizon $t_f$, time step $r$, input set $\mathcal{U}$ represented as a zonotope, factor $\lambda$,  desired zonotope order $\rho_d$.
	
	\textbf{Ensure:} Final reachable set $\mathcal{R}(t_f)$.
	
	}
	\begin{algorithmic}[1]
		\State $t_0 = 0,~ s = 0,~ \mathcal{R}(0) = \mathcal{X}_0,~ \mathcal{U}^{\Delta} = \mathcal{U} \oplus (-u^c)$
		\While{$t_s < t_f$}
			\State $w, A, B ,D ,E \leftarrow \operator{taylor}(z_s^*)$ \label{line:taylorSeries}
			\State $\mathcal{V}(t_s) = w \oplus B u^c \oplus \frac{1}{2} \operator{sq}(D,R(t_s) \times \mathcal{U})$ \label{line:quadMap}
			\Repeat \label{line:beginRepeat}
				\State $\overline{\Psi} (\tau _s) = $~\texttt{enlarge}$(\Psi(\tau_s),\lambda)$ \label{line:enlarge}
				\State $\Psi^{\Delta}(\tau_s) = \operator{abstrErr}(\mathcal{R}(t_s),\overline{\Psi}(\tau_s))$ \label{line:linErrorAlg}
				\State $\Psi(\tau_s) = \mathcal{V}(t_s) \oplus \Psi^{\Delta}(\tau_s)$ \label{line:linError}
				\hspace{8em}\smash{$\left.\rule{0pt}{6\baselineskip}\right\}\ \rotatebox[origin=c]{90}{$\mathcal{R}(t_{s+1}) = \operator{post}(\mathcal{R}(t_s),\Psi(\tau_s),z_s^*)$}$}
			\Until{$\Psi(\tau_s) \subseteq \overline{\Psi} (\tau _s)$} \label{line:endRepeat}
			\State $\mathcal{R}^{\Delta}(\tau_s) = \mathcal{R}^{p,\Delta}(\Psi^{\Delta}(\tau_s),r)$ \label{line:Rdelta}
			\State $\widehat{\mathcal{R}}(t_{s+1}) = \left( e^{Ar} \otimes \mathcal{R}(t_s) \boxplus \Gamma(r) \otimes \mathcal{V}(t_s) \right)$ \label{line:post}
			\Statex[4]$~\oplus \mathcal{R}^{\Delta}(\tau_s)$ 
			\State $\mathcal{R}(t_{s+1}) = \operator{reduce}(\widehat{\mathcal{R}}(t_{s+1}),\rho_d)$ \label{line:reduce}
			\State $t_{s+1} = t_s + r, ~ s:= s+ 1$
		\EndWhile
	\end{algorithmic}
\end{algorithm}

\begin{algorithm}[h!tb]
	\caption{$\Psi^{\Delta}(\tau_s) = \operator{abstrErr}(\mathcal{R}(t_s),\overline{\Psi}(\tau_s))$} \label{alg:linError}
	{\raggedright \textbf{Require:} Reachable set $\mathcal{R}(t_s)$, set of applied abstr. errors $\overline{\Psi}(\tau_s)$.
		
	\textbf{Ensure:} Set of abstraction errors $\Psi^{\Delta}(\tau_s)$.
	
	}
	\begin{algorithmic}[1]
			\State $\mathcal{Z}_z(t_s) = \operator{zono}(\mathcal{R}(t_s)) \times \mathcal{U}$ \label{line:zono}
				\State $\mathcal{R}_z^{\Delta}(\tau_s) = $~\texttt{post}$^{\Delta}(\mathcal{R}(t_s),\overline{\Psi} (\tau _s),A) \times \mathcal{U}$ \label{line:postDelta1}
				\State $\mathcal{R}_z^{\Delta}(\tau_s) = \operator{zono}(\mathcal{R}_z^{\Delta}(\tau_s))$ \label{line:postDelta2}
				\State $\mathcal{V}^{\Delta}(\tau_s) = $~\texttt{varInputs}$(\mathcal{Z}_z(t_s),\mathcal{R}_z^{\Delta}(\tau_s),\mathcal{U}^{\Delta},B,D)$ \label{line:varInput}
				\State $\mathcal{R}(\tau_s) = \mathcal{R}(t_s) \oplus \mathcal{R}_z^{\Delta}(\tau_s)$ \label{line:reachSetInt}
				\State $\mathcal{L}(\tau_s) = $~\texttt{lagrangeRemainder}$(\mathcal{R}(\tau_s),E,z^*)$ \label{line:lagRem}
				\State $\Psi^{\Delta}(\tau_s) =  \mathcal{V}^{\Delta}(\tau_s) \oplus \mathcal{L}(\tau_s)$ \label{line:linErrorSet}
	\end{algorithmic}
\end{algorithm}

The accuracy of the obtained reachable sets is almost entirely determined by the reachable sets $\mathcal{R}(t_s)$ at points in time $t_s$ since only these sets are propagated \cite{Althoff2013a}, while the reachable sets of time intervals $\mathcal{R}(\tau_s)$ only fill the time gaps. Consequently, subsequent derivations focus on the reachable sets at points in time $\mathcal{R}(t_s)$.

\subsection{Reachable Subsets}

In this section we show how to efficiently obtain reachable subsets within a pre-computed reachable set as presented in Fig.~\ref{fig:ExampleReach}. The main idea is illustrated in Example~\ref{ex:reach}, followed by Lemmas leading to the main result in Theorem~\ref{theo:main}.

\begin{figure}
\begin{center}
	\setlength{\belowcaptionskip}{-10pt}
	\includegraphics[width = 0.45 \textwidth]{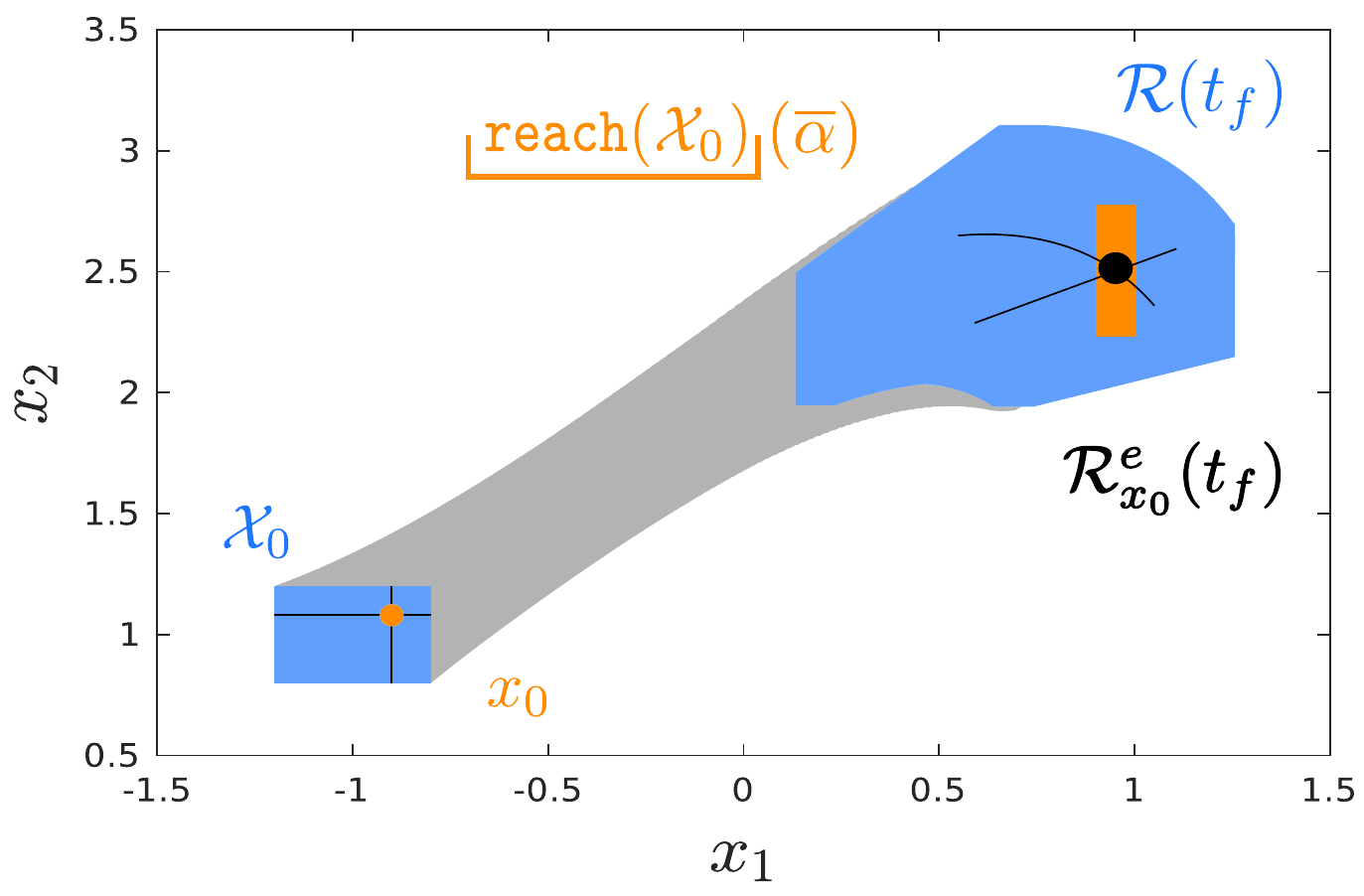}
	\caption{Visualization of the reachable set from Example~\ref{ex:reach}. The black lines correspond to the points where $\mathbf{\alpha_1 = 0.5 = const.}$ and $\mathbf{\alpha_2 = 0.4 = const.}$ holds.}
	\label{fig:ExampleReach}
	\end{center}
\end{figure}

\begin{example}
	We consider the following two-dimensional system
	\begin{equation}
		\begin{split}
			& \dot x_1 = x_2 \\
			& \dot x_2 = (1 - x_1^2) x_2 - x_1,
		\end{split}
	\end{equation}
	with the initial set
	\begin{equation}
		\mathcal{X}_0 = \bigg \{ \begin{bmatrix} -1 \\ 1 \end{bmatrix} + \begin{bmatrix} 0.2 \\ 0 \end{bmatrix} \alpha_1 + \begin{bmatrix} 0 \\ 0.2 \end{bmatrix} \alpha_2 ~ \bigg | ~ \alpha_1, \alpha_2 \in [-1,1] \bigg \},
		\label{eq:initialSet}
	\end{equation}
	and the initial point $x_0 = [-0.9~1.08]^T \in \mathcal{X}_0$. The initial point $x_0$ can be parameterized with $\overline{\alpha} = [0.5~0.4]^T$, so that $x_0 = \eval{\mathcal{X}_0}(\overline{\alpha})$. 
	The computation of the reachable set for a time horizon of $t_f = 1s$ with Alg.~\ref{alg:reach} yields
	\begin{equation}
		\begin{split}
		\mathcal{R}(t_f) = \bigg \{ & \begin{bmatrix} 0.73 \\ 2.52 \end{bmatrix} + \begin{bmatrix} 0.25 \\ -0.1 \end{bmatrix} \alpha_1 + \begin{bmatrix} 0.26 \\ 0.2 \end{bmatrix} \alpha_2   \\
		& + \begin{bmatrix} -0.04 \\ -0.09 \end{bmatrix} \alpha_1^2 - \begin{bmatrix} 0 \\ 0.1 \end{bmatrix} \alpha_1 \alpha_2 + \begin{bmatrix} 0.05 \\ 0 \end{bmatrix} \beta_1  \\
		& + \begin{bmatrix} 0 \\ 0.27 \end{bmatrix} \beta_2 ~ \bigg | ~ \alpha_1,\alpha_2, \beta_1,\beta_2 \in [-1,1] \bigg \}.
		\end{split}
		\label{eq:reachSet}
	\end{equation}
	As visualized in Fig.~\ref{fig:ExampleReach}, the exact reachable set $\mathcal{R}_{x_0}^e(t_f)$ for the initial point $x_0$ can be enclosed by evaluating \eqref{eq:reachSet} for the parameter values $\overline{\alpha} = [0.5~0.4]^T$ corresponding to the initial point $x_0$.
	\label{ex:reach}
\end{example}

We now prove the correctness of the concept demonstrated by Example~\ref{ex:reach}. Let us start with the computation of the abstraction error returned by Alg.~\ref{alg:linError}:

\begin{lemma}
	Given the reachable sets $\mathcal{R}^{(1)}(t_s)$,$\mathcal{R}^{(2)}(t_s) \subset \Rn$ with $\mathcal{R}^{(1)}(t_s) \subseteq \mathcal{R}^{(2)}(t_s)$ and the sets of applied abstraction errors $\overline{\Psi}^{(1)}(\tau_s)$, $\overline{\Psi}^{(2)}(\tau_s) \subset \Rn$ with $\overline{\Psi}^{(1)}(\tau_s) \subseteq \overline{\Psi}^{(2)}(\tau_s)$, it holds that
	\begin{equation}
		\begin{split}
			& \Psi^{\Delta (1)}(\tau_s) \subseteq \Psi^{\Delta (2)}(\tau_s) \\
			& with ~~ \Psi^{\Delta (1)}(\tau_s) = \operator{abstrErr}\big(\mathcal{R}^{(1)}(t_s),\overline{\Psi}^{(1)}(\tau_s)\big) \\
			& ~~~~~~~~ \Psi^{\Delta (2)}(\tau_s) = \operator{abstrErr}\big(\mathcal{R}^{(2)}(t_s),\overline{\Psi}^{(2)}(\tau_s)\big),
		\end{split}
	\end{equation}
	so that \operator{abstrErr} as defined by Alg.~\ref{alg:linError} satisfies \eqref{eq:binary}.
	\label{lemma:linError}
\end{lemma}

\begin{proof}
	As visualized in Fig.~\ref{fig:algorithmLinError}, Alg.~\ref{alg:linError} is a composition of unary set operations that satisfy \eqref{eq:unary} and binary set operations that satisfy \eqref{eq:binary}. Therefore, it holds according to Lemma~\ref{lemma:conseqSubset} that \operator{abstrErr} as defined by Alg.~\ref{alg:linError} satisfies \eqref{eq:binary}. \hfill ~
\end{proof}

\begin{figure}[h]
\begin{center}
	\includegraphics[width = 0.45 \textwidth]{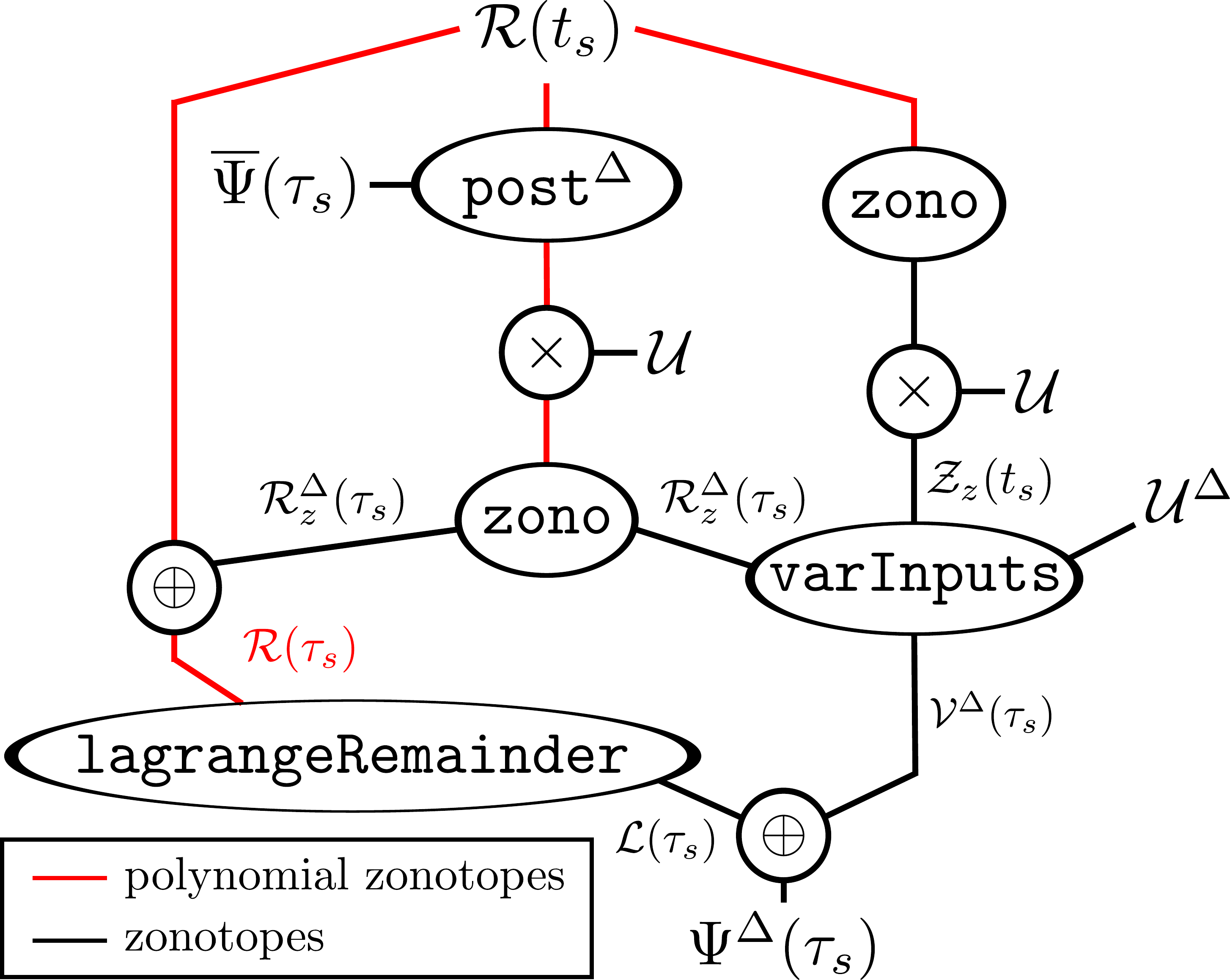}
	\caption{Flow chart for Alg.~\ref{alg:linError}.}
	\label{fig:algorithmLinError}
	\end{center}
\end{figure}

 Using the dependency-preserving properties of operations, we show that Alg.~\ref{alg:reach} is dependency-preserving. We start with the \operator{post} operator, i.e., a single iteration of Alg.~\ref{alg:reach}:

\begin{lemma}
	Given a reachable set $\mathcal{R}(t_s) \subset \Rn$ represented as a polynomial zonotope, it holds that
	\begin{equation}
		\begin{split}
			\forall \alpha \in \mathcal{I}_m: ~~ & \operator{post}\left( \eval{\mathcal{R}(t_s)}(\alpha),\Psi(\tau_s),z_s^* \right) \\ & \subseteq \eval{\operator{post}(\mathcal{R}(t_s),\mathbf{0},z_s^*)}(\alpha),
		\end{split}
		\label{eq:lemmaStep}
	\end{equation}
	where $\Psi(\tau_s)$ is the set of abstraction errors resulting from the computation of $\operator{post}(\mathcal{R}(t_s),\mathbf{0},z_s^*)$ and $z_s^*$ is the expansion point of the Taylor series.
	\label{lemma:step}
\end{lemma}

\begin{proof}
	A flow chart for the \operator{post} operation in Alg.~\ref{alg:reach} is shown in Fig.~\ref{fig:algorithmReach}. First, we consider the operations \operator{A}, \operator{B}, and \operator{C} illustrated by the blocks in Fig.~\ref{fig:algorithmReach}:
	
	~
	
	\noindent \textbf{Operation \operator{A}:}
	As visualized in Fig.~\ref{fig:algorithmReach}, operation $\operator{A}$ is defined by a composition of dependency-preserving operations so that \operator{A} is dependency-preserving according to Lemma~\ref{lemma:conseqExecution}.
	
	~
	
	\noindent \textbf{Operation \operator{B}:}
	For the repeat-until-loop in Lines~\ref{line:beginRepeat}-\ref{line:endRepeat} of Alg.~\ref{alg:reach} it is sufficient to consider only the values from the last iteration since only these are applied in subsequent computations. The operation \operator{abstrErr} satisfies \eqref{eq:binary} according to Lemma~\ref{lemma:linError}. As visualized in Fig.~\ref{fig:algorithmReach}, the binary operation $\operator{B}$ is defined by a composition of unary operations satisfying \eqref{eq:unary} and binary operations satisfying \eqref{eq:binary}, so that \operator{B} satisfies \eqref{eq:binary} according to Lemma~\ref{lemma:conseqSubset}.
	
	~
	
	\noindent \textbf{Operation \operator{C}:}
	According to Fig.~\ref{fig:algorithmReach}, operation \operator{C} is defined as 
	\begin{equation}
		\operator{C}(\mathcal{R}(t_s)) = \operator{A}(\mathcal{R}(t_s)) \oplus \operator{B}(\mathcal{R}(t_s),\Psi(\tau_s)).
		\label{eq:C}
	\end{equation}
	Since \operator{A} is dependency-preserving, it holds according to \eqref{eq:depPreserving} that
	\begin{equation}
		\forall \alpha \in \mathcal{I}_m: ~~ \operator{A}\left( \eval{\mathcal{R}(t_s)}(\alpha) \right) \subseteq \eval{\operator{A}(\mathcal{R}(t_s))}(\alpha).
		\label{eq:proofLemma1}
	\end{equation}
	Since $\eval{\mathcal{R}(t_s)}(\alpha) \subseteq \mathcal{R}(t_s)$ and \operator{B} satisfies \eqref{eq:binary},
	\begin{equation}
		\forall \alpha \in \mathcal{I}_m: ~~ \operator{B}\left( \eval{\mathcal{R}(t_s)}(\alpha),\Psi(\tau_s) \right) \subseteq \operator{B}(\mathcal{R}(t_s),\Psi(\tau_s)).
		\label{eq:proofLemma2}
	\end{equation}
	Furthermore, since \operator{B} returns a zonotope that does not contain the polynomial zonotope factors $\alpha$,
	\begin{equation}
		\begin{split}
			\forall \alpha \in \mathcal{I}_m: ~~ & \eval{\operator{A}(\mathcal{R}(t_s))}(\alpha) \oplus \operator{B}(\mathcal{R}(t_s),\Psi(\tau_s)) \\
			& = \eval{\operator{A}(\mathcal{R}(t_s)) \oplus \operator{B}(\mathcal{R}(t_s),\Psi(\tau_s))}(\alpha).
		\end{split}
		\label{eq:proofLemma3}
	\end{equation}
	Consequently,
	\begin{equation*}
		\begin{split}
		& \forall \alpha \in \mathcal{I}_m: \\
		& \operator{C}\left(\eval{\mathcal{R}(t_s)}(\alpha)\right) \overset{\eqref{eq:C}}{=}  \operator{A}\left(\eval{\mathcal{R}(t_s)}(\alpha)\right) \oplus \operator{B}\left(\eval{\mathcal{R}(t_s)}(\alpha),\Psi(\tau_s)\right) \\ & \overset{\eqref{eq:proofLemma1}}{\subseteq} \eval{\operator{A}(\mathcal{R}(t_s))}(\alpha) \oplus \operator{B}\left(\eval{\mathcal{R}(t_s)}(\alpha),\Psi(\tau_s)\right) \\
		& \overset{\eqref{eq:proofLemma2}}{\subseteq} \eval{\operator{A}(\mathcal{R}(t_s))}(\alpha) \oplus \operator{B}(\mathcal{R}(t_s),\Psi(\tau_s)) \\
		& \overset{\eqref{eq:proofLemma3}}{=} \eval{\operator{A}(\mathcal{R}(t_s)) \oplus \operator{B}(\mathcal{R}(t_s),\Psi(\tau_s))}(\alpha) \overset{\eqref{eq:C}}{=} \eval{\operator{C}(\mathcal{R}(t_s))}(\alpha),
		\end{split}	
	\end{equation*}
	which proves that operation \operator{C} is dependency-preserving (see \eqref{eq:depPreserving}).
	
	As visualized in Fig.~\ref{fig:algorithmReach}, the \operator{post} operation is defined by the composition of dependency-preserving operations resulting in a dependency-preserving operation according to Lemma~\ref{lemma:conseqExecution}. \hfill ~
\end{proof}	
	
\begin{figure}
\begin{center}
	\includegraphics[width = 0.45 \textwidth]{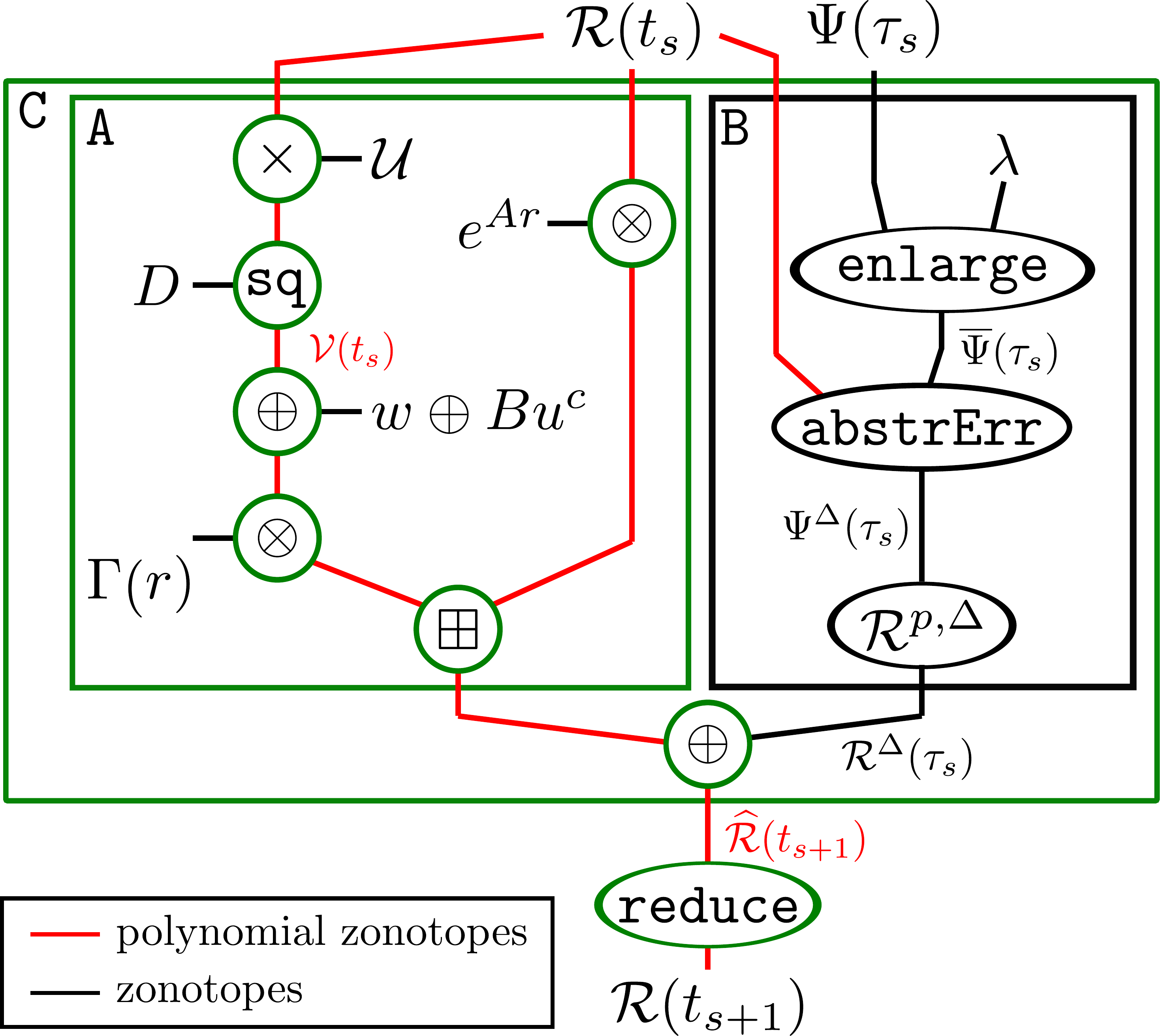}
	\caption{Flow chart for Alg.~\ref{alg:reach}. Dependency-preserving operations are marked in green.}
	\label{fig:algorithmReach}
	\end{center}
\end{figure}

Using the result for a single iteration, we now prove that Alg.~\ref{alg:reach} is dependency-preserving:

\begin{lemma}
	Given an initial set $\mathcal{X}_0 \subset \Rn$ represented as a polynomial zonotope, it holds that
	\begin{equation}
	\begin{split}
			\forall \alpha \in \mathcal{I}_m: ~~ & \operator{reach}\left( \eval{\mathcal{X}_0}(\alpha),\Psi(\tau_s),z_s^* \right) \\
			& \subseteq \eval{\operator{reach}(\mathcal{X}_0,\mathbf{0},z_s^*)}(\alpha)
	\end{split}
		\label{eq:lemmaReach}
	\end{equation}
	where $\Psi(\tau_s)$, $s = 1,\dots,N$ are the sets of abstraction errors resulting from the computation of $\operator{reach}(\mathcal{R}(t_s),\mathbf{0},z_s^*)$ in Alg.~\ref{alg:reach}, $z_s^*$, $s = 1,\dots,N$ are the expansion points of the Taylor series, $N = \lceil \frac{t_f}{r} \rceil$ is the number of time steps, $t_f$ is the time horizon, and $r$ is the time step size. 
	\label{lemma:reach}
\end{lemma}

\begin{proof}
	The \operator{reach} operation as defined by Alg.~\ref{alg:reach} can be expressed equivalently as 		    
	\begin{equation*}
		\mathcal{R}(t_f) = \operator{reach}(\mathcal{X}_0,\Psi(\tau_s),z_s^*) = \operator{post}_N(\mathcal{X}_0,\Psi(\tau_s),z_s^*),
	\end{equation*}
	where $\operator{post}_N$ denotes the $N$-times consecutive composition of the \operator{post} operation
	\begin{equation*}
		\mathcal{R}(t_f) = \underbrace{\operator{post}( \dots \operator{post}(\mathcal{X}_0,\Psi(\tau_1),z_1^*) \dots ,\Psi(\tau_N),z_N^*)}_{\operator{post}_N(\mathcal{X}_0,\Psi(\tau_s),z_s^*)}.
	\end{equation*} 
	Then, since the \operator{post} operation is dependency-preserving (see \linebreak[4]Lemma~\ref{lemma:step}) and the composition of dependency-preserving operations yields a dependency-preserving operation (see \linebreak[4]Lemma~\ref{lemma:conseqExecution}), operation $\operator{reach}(\mathcal{X}_0,\Psi(\tau_s),z_s^*) = \operator{post}_N(\mathcal{X}_0,\linebreak[4]\Psi(\tau_s),z_s^*)$ is dependency-preserving. \hfill ~
\end{proof}

Finally, we formulate the main result:

\begin{theorem}
	Given an initial set $\mathcal{X}_0 \subset \Rn$ represented as a polynomial zonotope,
	\begin{equation}
	\begin{split}
	\forall \alpha \in \mathcal{I}_m: ~~ & \mathcal{R}_{x_0}^e(t_f) \subseteq \eval{\operator{reach}(\mathcal{X}_0,\mathbf{0},z_s^*)}(\alpha) \\
	& \mathrm{with} ~~ x_0 = \eval{\mathcal{X}_0}(\alpha),
	\end{split}
	\label{eq:mainTheorem}
\end{equation}
	where $z_s^*$ and \operator{reach} are as in Lemma~\ref{lemma:reach}.
	 \label{theo:main}
\end{theorem}

\begin{proof}
	As shown in \cite{Althoff2013a}, Alg.~\ref{alg:reach} computes an over-\linebreak[4]approximation of the exact reachable set as follows:
	\begin{equation}
		\mathcal{R}_{x_0}^e(t_f) \subseteq \operator{reach}\left(\eval{\mathcal{X}_0}(\alpha),\Psi(\tau_s),z_s^* \right),
		\label{eq:subsetExactReach}
	\end{equation}
	where $\Psi(\tau_s)$, $s = 1,\dots,N$ are the sets of abstraction errors resulting from the computation of $\operator{reach}(\mathcal{X}_0,\mathbf{0},z_s^*)$ so that according to Lemma~\ref{lemma:reach} we have:
	\begin{equation*}
	\begin{split}
			\forall \alpha \in \mathcal{I}_m: ~~ & \mathcal{R}_{x_0}^e(t_f) \overset{\eqref{eq:subsetExactReach}}{\subseteq} \operator{reach}\big(\eval{\mathcal{X}_0}(\alpha),\Psi(\tau_s),z_s^*\big) \\
			& \subseteq \eval{\operator{reach}(\mathcal{X}_0,\mathbf{0},z_s^*)}(\alpha),
	\end{split}
	\end{equation*}
	which concludes the proof. \hfill ~
\end{proof}

Since Theorem~\ref{theo:main} holds for all points $x_0 \in \mathcal{X}_0$ inside the initial set $\mathcal{X}_0$, it is obvious that Theorem~\ref{theo:main} equally holds for all sets $\widehat{\mathcal{X}}_0 \subseteq \mathcal{X}_0$. Furthermore, Theorem~\ref{theo:main} also holds at intermediate time steps $\mathcal{R}(t_s)$, $s = 1,\dots,N$ with $N = \lceil \frac{t_f}{r} \rceil$. Since the time interval reachable set $\mathcal{R}(\tau_s)$ is computed by adding uncertainty to the time point reachable set $\mathcal{R}(t_s)$ (see Line~\ref{line:reachSetInt} of Alg.~\ref{alg:linError}), Theorem~\ref{theo:main} equally holds for the reachable set $\mathcal{R}(\tau_s)$ of time intervals $\tau_s = [t_{s},t_{s+1}]$.

\subsection{Computational Complexity}

We now derive the computational complexity for extracting a reachable subset from the final reachable set according to \eqref{eq:evalFunPolyZono}, where
\begin{equation}
	f_{G,E}(\alpha) = \sum_{i=1}^h \bigg ( \underbrace{ \prod_{k=1}^m \alpha_k^{E_{(k,i)}} }_{P_i} \bigg ) G_{(\cdot,i)}.
	\label{eq:complexity}
\end{equation}
Computation of each $P_i$ in \eqref{eq:complexity} requires $m$ exponentiations and $m-1$ multiplications, and computation of $P_i G_{(\cdot,i)}$ requires $n$ multiplications. Since there are $h$ terms $P_i G_{(\cdot,i)}$ in \eqref{eq:complexity}, computation of all terms requires $h(2m+n-1)$ operations. The computation of the outer sum in \eqref{eq:complexity} requires $n(h-1)$ additions, so that the computation of $f_{G,E}(\alpha)$ requires in total $h(2m+n-1) + n(h-1)$ operations. It holds for the number of dependent factors $m$ and the number of dependent generators $h$ that $m = c_m n$ and $h = c_h n$ with $c_m,c_h \in \R_{\geq 0}$. The complexity for the extraction of a reachable subset is therefore $\mathcal{O}(h(2p+n-1) + n(h-1)) = \mathcal{O}(n^2)$. Since the computation of the reachable set with the conservative polynomialization approach has complexity $\mathcal{O}(n^5)$ \cite{Althoff2013a}, our novel extraction of reachable subsets is computational much more efficient.


\section{Applications}
\label{sec:Applications}

There are a lot of different applications for the result presented in this paper. In this section, we introduce some of the possible applications and demonstrate their efficiency using numerical examples. However, due to limited space, we only present the general concepts and omit implementation details. The implementation of our approach including all numerical examples will be made publicly available in the next release of the CORA toolbox \cite{Althoff2015a}. Also, we carried out all computations in MATLAB on a 2.9GHz quad-core i7 processor with 32GB memory.

\begin{table}
\begin{center}
\caption{Comp. times for the numerical examples in [s].}
\label{tab:compTime}
\begin{tabular}{ l c c c }
 \toprule
 \multirow{2}{2cm}{\textbf{Application}} & \multicolumn{2}{c}{\textbf{Computation Time in [s]}} \\ \cmidrule{2-3}
		& \textbf{New Approach} & \textbf{Prev. Solution} \\ \midrule 
 Falsification (2D) & $0.12$ & $0.38$ \\
 Falsification (12D) & $0.13$ & $10.4$ \\ 
 Optimization & $0.13$ & $172$ \\
 Safe Control & $0.06$ & $21.4$ \\
 \bottomrule 
\end{tabular}
\end{center}
\end{table}

\subsection{Falsifying States}

Let us consider a specification defined by a linear inequality constraint $a^T x \leq b$, $a \in \Rn,b \in \R$. If the final reachable set $\mathcal{R}(t_f)$ for the initial set $\mathcal{X}_0$ does not fulfill the specification, it would be helpful to know which initial states result in a violation. According to Theorem~\ref{theo:main}, the states inside the set $\mathcal{S} \subseteq \mathcal{X}_0$ defined as 
\begin{equation*}
	\begin{split}
		& \mathcal{S} = \bigcup_{\alpha \in \mathcal{B}} \eval{\mathcal{X}_0}(\alpha) \\
		& \mathrm{with} ~~ \mathcal{B} = \left \{ \alpha \in \mathcal{I}_m ~ | ~ a^T \otimes \eval{ \mathcal{R}(t_f)}(\alpha) \leq b \right \}
	\end{split}
\end{equation*}
are guaranteed to fulfill the specification. The set of states $\mathcal{F} \subseteq \mathcal{X}_0$ that potentially result in the violation of the specification can consequently be computed as $\mathcal{F} = \mathcal{X}_0 \setminus \mathcal{S}$.

Theorem~\ref{theo:main} is used to efficiently determine an initial point as well as a suitable input trajectory falsifying the specification. The initial point $x^* \in \mathcal{X}_0$ that results in the largest violation of the specification can be computed by solving the optimization problem
\begin{equation}
	\begin{split}
		& x^* = \eval{\mathcal{X}_0}(\alpha^*) \\
		& \mathrm{with} ~~ \alpha^* = \argmax_{\alpha \in \mathcal{I}_p} ~~ a^T \otimes \eval{\mathcal{R}(t_f)}(\alpha).
	\end{split}
	\label{eq:falsification}
\end{equation}
Since this reduces the \textit{safety falsification} task to finding a suitable input trajectory, falsifying trajectories can be found much more efficiently as shown by an example:

\begin{figure}
\begin{center}
	\includegraphics[width = 0.45 \textwidth]{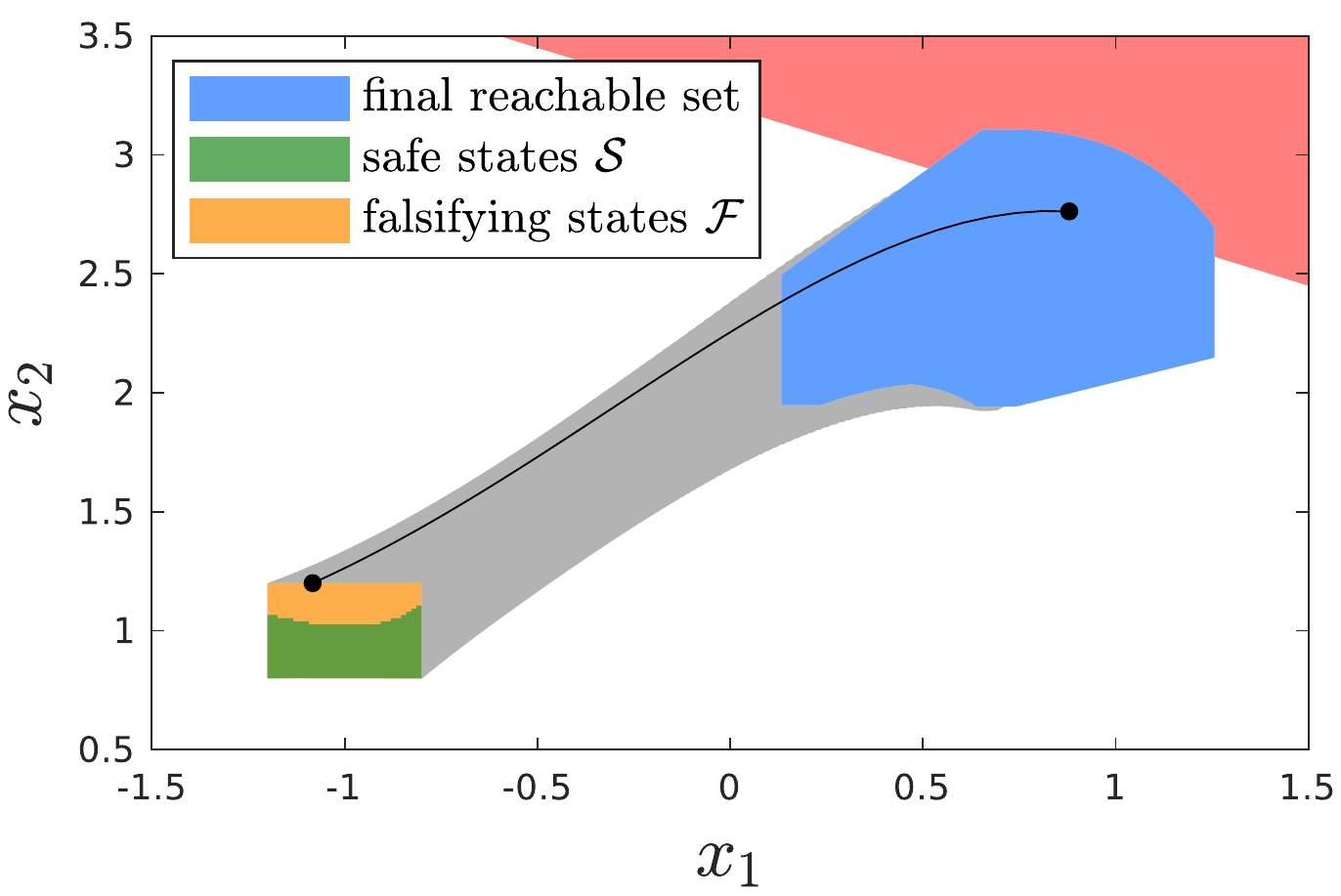}
	\caption{Visualization of Example~\ref{ex:falsification}. The determined falsifying trajectory is shown in black.}
	\label{fig:ExampleFalsification}
	\end{center}
\end{figure}

\begin{example}
	We consider the reachability problem from Example~\ref{ex:reach} in combination with the specification $x_1 + 2 x_2 \leq 6.4$. As shown in Fig.~\ref{fig:ExampleFalsification}, the final reachable set violates the constraint, where in green the safe states $\mathcal{S}$ and in orange the falsifying states $\mathcal{F}$ are visualized. Since the reachability problem from Example~\ref{ex:reach} does not include uncertain inputs, a falsifying trajectory can simply be determined by solving the optimization problem in \eqref{eq:falsification}, which takes $0.12$ seconds (see Table~\ref{tab:compTime}). As a comparison, we also determined a falsifying trajectory using the simulated annealing algorithm of the falsification toolbox S-TALIRO {\normalfont \cite{Annapureddy2011}}. Since the simulated annealing algorithm is non-deterministic we computed the average computation time from $10$ executions, which results in the value $0.38$ seconds (see Table~\ref{tab:compTime}).
	\label{ex:falsification}
\end{example}

We demonstrate the scalability of our approach with the system dimension by a second example:

\begin{example}
	We consider the 12-dimensional quadrotor benchmark from the ARCH 19 competition {\normalfont \cite{ARCH19nonlinear}} in combination with the specification $x_3 \leq 1.355 m ~ \forall t \in [0s,5s]$ (see Fig.~\ref{fig:ExampleFalsification12D}). Since the benchmark does not include uncertain inputs, a falsifying trajectory can be computed by solving \eqref{eq:falsification}, which takes $0.13$ seconds (see Table~\ref{tab:compTime}). The computation time of the simulated annealing algorithm from S-TALIRO averaged over $10$ executions is $10.4$ seconds (see Table~\ref{tab:compTime}).
	\label{ex:falsification12D}
\end{example}

\begin{figure}
\begin{center}
	\includegraphics[width = 0.44 \textwidth]{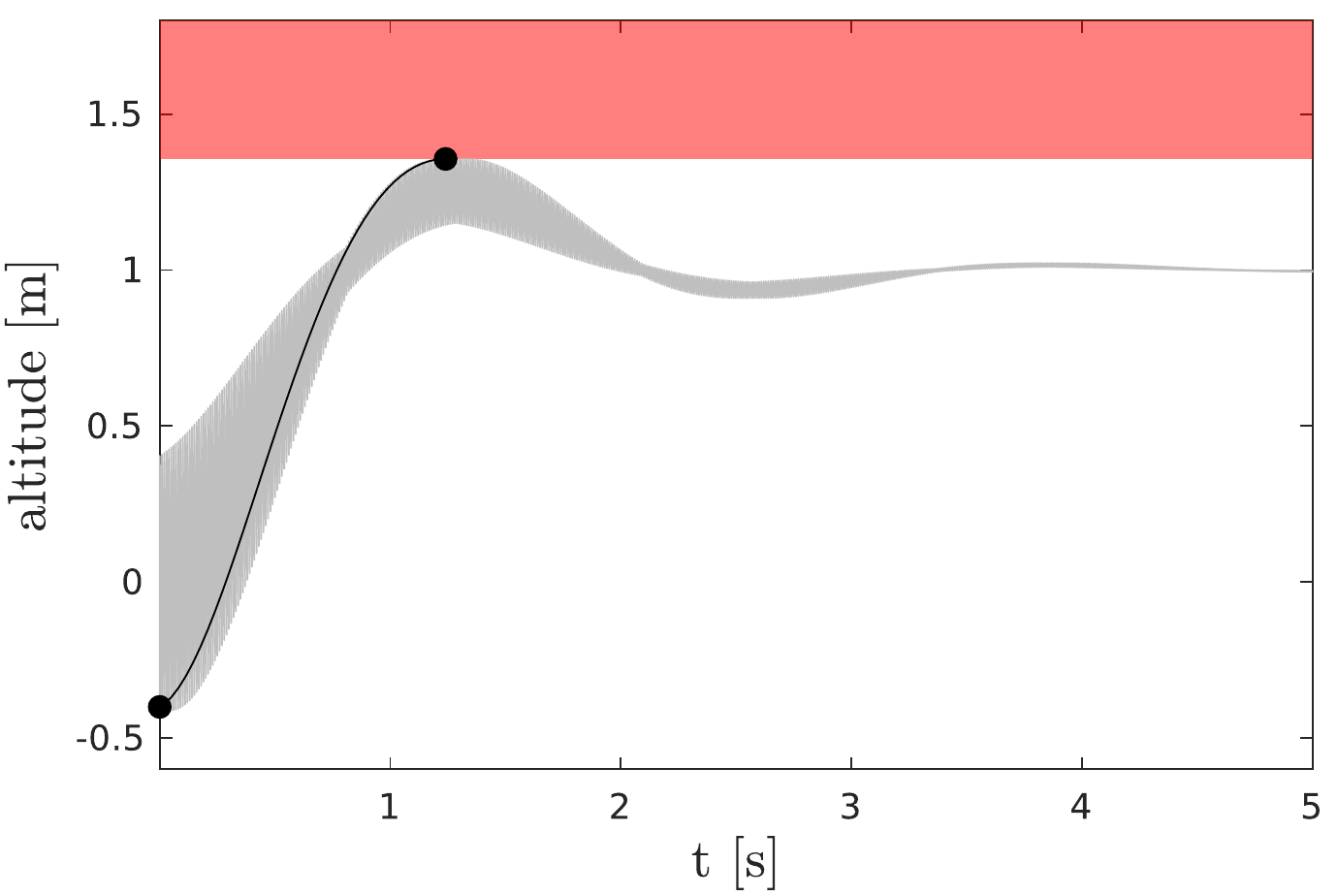}
	\caption{Visualization of Example~\ref{ex:falsification12D}. The determined falsifying trajectory is shown in black.}
	\label{fig:ExampleFalsification12D}
	\end{center}
\end{figure}

\subsection{Optimization over Reachable Sets}

Since reachability analysis is computational expensive, optimization over reachable sets is often infeasible. However, with our new approach it is possible to achieve major speed-ups for optimizing over reachable sets:

\begin{figure}
\begin{center}
	\includegraphics[width = 0.45 \textwidth]{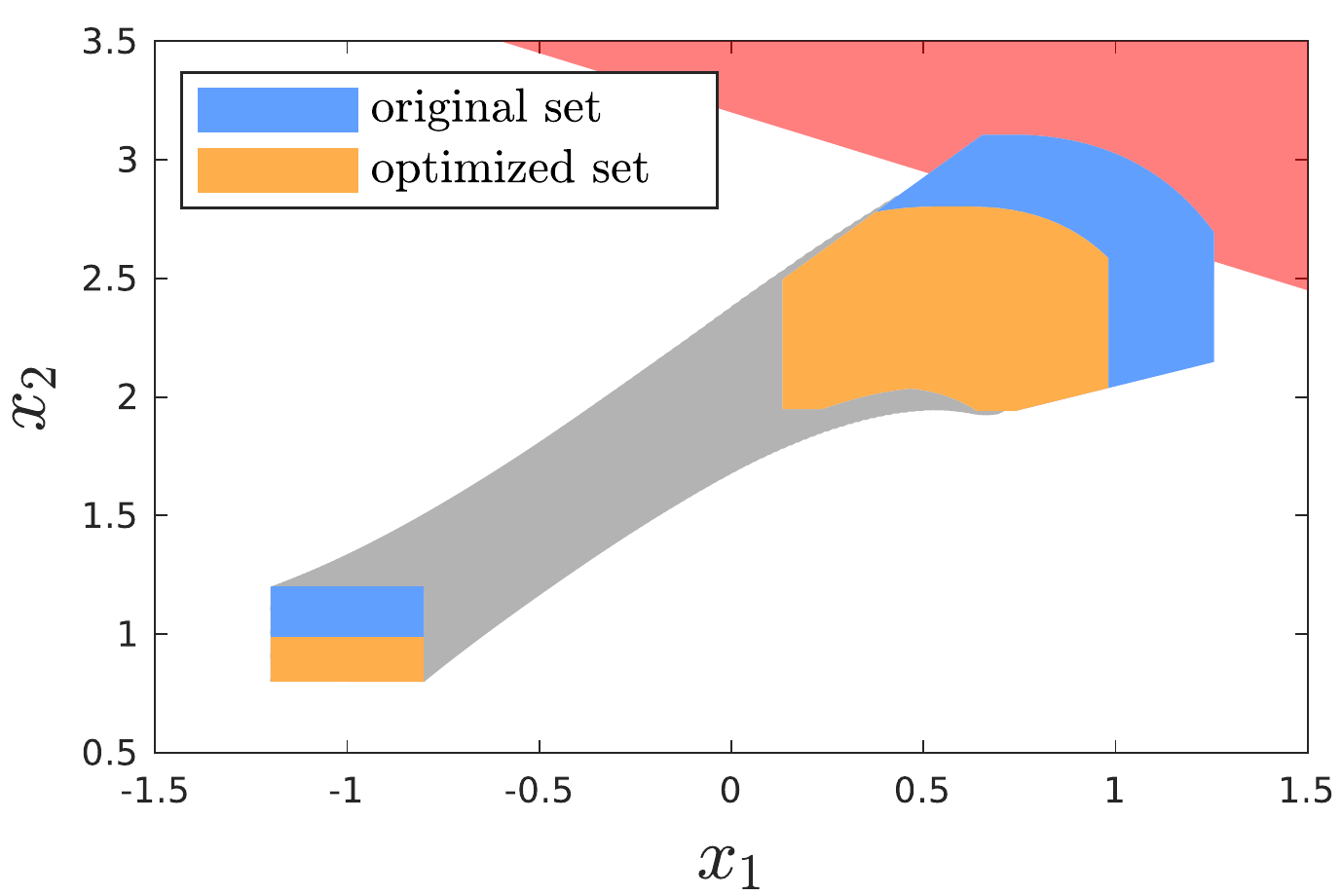}
	\caption{Visualization of Example~\ref{ex:optimization}.}
	\label{fig:ExampleOptimization}
	\end{center}
\end{figure}

\begin{example}
	We consider the reachability analysis problem from Example~\ref{ex:reach} and the inequality constraint $[1~2]~x \leq 6.4$. As shown in Fig.~\ref{fig:ExampleOptimization}, the final reachable set violates the constraint. We want to determine the modified initial set $\mathcal{X}_0^* \subset \mathcal{X}_0$ with the maximum volume for which the final reachable set satisfies the constraint, which can be formulated as an optimization problem: We use the upper and lower bounds $\underline{\alpha}_1,\overline{\alpha}_1,\underline{\alpha}_2,\overline{\alpha}_2$ for the factors $\alpha_1$ and $\alpha_2$ as the variables for the optimization problem. Consequently, the initial set from Example~\ref{ex:reach} becomes 
	\begin{equation*}
		\begin{split}
			\mathcal{X}_0(\underline{\alpha}_1,\overline{\alpha}_1,\underline{\alpha}_2,\overline{\alpha}_2) = \bigg \{ & \begin{bmatrix} -1 \\ 1 \end{bmatrix} + \begin{bmatrix} 0.2 \\ 0 \end{bmatrix} \alpha_1 + \begin{bmatrix} 0 \\ 0.2 \end{bmatrix} \alpha_2 ~ \bigg | \\
			& \alpha_1 \in [\underline{\alpha}_1,\overline{\alpha}_1],~ \alpha_2 \in [\underline{\alpha}_2,\overline{\alpha}_2] \bigg \},
		\end{split}
	\end{equation*}
	 where we denote the dependence of the initial set on the bounds $\underline{\alpha}_1,\overline{\alpha}_1,\underline{\alpha}_2,\overline{\alpha}_2$ by $\mathcal{X}_0(\underline{\alpha}_1,\overline{\alpha}_1,\underline{\alpha}_2,\overline{\alpha}_2)$. The optimization problem can be formulated as
	 \begin{equation}
	 	\begin{split}
	 		& \max_{\underline{\alpha}_1,\overline{\alpha}_1,\underline{\alpha}_2,\overline{\alpha}_2} ~ \operator{volume}\big(\mathcal{X}_0(\underline{\alpha}_1,\overline{\alpha}_1,\underline{\alpha}_2,\overline{\alpha}_2)\big) \\
	 		& ~ \\
	 		& ~~~~~ s.t. ~ [1~2] \otimes \operator{reach}\big(\mathcal{X}_0(\underline{\alpha}_1,\overline{\alpha}_1,\underline{\alpha}_2,\overline{\alpha}_2),\mathbf{0},z_s^*\big) \leq 6.4,
	 	\end{split}
	 	\label{eq:optProblem}
	 \end{equation}
	 where the operator $\operator{volume}(\mathcal{S})$ returns the volume of a set $\mathcal{S} \subset \Rn$. The optimization problem \eqref{eq:optProblem} is hard to solve since each evaluation of the constraint function requires the computationally expensive execution of the reachability algorithm. However, if we exploit the dependency preservation between the initial states and the reachable states introduced in this paper, the constraint can be equivalently formulated as 
	 \begin{equation}
	 	\max_{\substack{\alpha_1 \in [\underline{\alpha}_1,\overline{\alpha}_1] \\ \alpha_2 \in [\underline{\alpha}_2,\overline{\alpha}_2]}} ~[1~2] \otimes \eval{\mathcal{R}(t_f)}\big([\alpha_1~\alpha_2]^T \big) \leq 6.4,
	 	\label{eq:ineqConstraint}
	 \end{equation}
	 where $\mathcal{R}(t_f)$ is the final reachable set (see \eqref{eq:reachSet}). Inserting the numerical values for $\mathcal{R}(t_f)$ from \eqref{eq:reachSet} into \eqref{eq:ineqConstraint} yields
	 \begin{equation*}
	 	\max_{\substack{\alpha_1 \in [\underline{\alpha}_1,\overline{\alpha}_1] \\ \alpha_2 \in [\underline{\alpha}_2,\overline{\alpha}_2]}}~ 0.05 \alpha_1 + 0.66 \alpha_2 - 0.22 \alpha_1^2 - 0.2 \alpha_1 \alpha_2 \leq 0.04.
	 \end{equation*}	
	 The volume of the initial set can be computed as 
	 \begin{equation*}
	 	\operator{volume}\big(\mathcal{X}_0(\underline{\alpha}_1,\overline{\alpha}_1,\underline{\alpha}_2,\overline{\alpha}_2)\big) = 0.04 \left(\overline{\alpha}_1 - \underline{\alpha}_1 \right) \left(\overline{\alpha}_2 - \underline{\alpha}_2 \right),
	 \end{equation*}	
	 which simplifies the optimization problem \eqref{eq:optProblem} to
	 \begin{equation}
	 	\begin{split}
	 		& \max_{\underline{\alpha}_1,\overline{\alpha}_1,\underline{\alpha}_2,\overline{\alpha}_2} ~ 0.04 \left(\overline{\alpha}_1 - \underline{\alpha}_1 \right) \left(\overline{\alpha}_2 - \underline{\alpha}_2 \right) \\
	 		& ~ \\
	 		& s.t. \max_{\substack{\alpha_1 \in [\underline{\alpha}_1,\overline{\alpha}_1] \\ \alpha_2 \in [\underline{\alpha}_2,\overline{\alpha}_2]}}~ 0.05 \alpha_1 + 0.66 \alpha_2 - 0.22 \alpha_1^2 - 0.2 \alpha_1 \alpha_2 \leq 0.04.
	 	\end{split}
	 	\label{eq:optProblemSimple}
	 \end{equation}
	 The solution for the optimization problem \eqref{eq:optProblemSimple} is visualized in Fig.~\ref{fig:ExampleOptimization}. As shown in Table~\ref{tab:compTime} solving the optimization problem \eqref{eq:optProblem} using MATLABs \texttt{fmincon} with the \texttt{interior-point} algorithm takes $221$ seconds, but solving the simplified optimization problem \eqref{eq:optProblemSimple} only takes $0.1$ seconds. 
	 \label{ex:optimization}
\end{example}

\subsection{Motion-Primitive Based Control}

In this section, we consider a scenario where a maneuver automaton is used to control a system. The approach in \cite{Schuermann2017a} constructs a provably safe maneuver automaton by using reachability analysis. In particular, for each motion-primitive of the maneuver automaton, the reachable set for an initial set $\mathcal{X}_0$ is computed offline. During online application our novel approach can directly extract the reachable set for a measured state $\widehat{x}$ from the offline-computed reachable set. Generally, since the reachable set for $\widehat{x}$ is much smaller than the reachable set for $\mathcal{X}_0$, planning with the reachable set for $\widehat{x}$ greatly reduces the conservatism as demonstrated with a numerical example:

\begin{figure}
\begin{center}
	\includegraphics[width = 0.48 \textwidth]{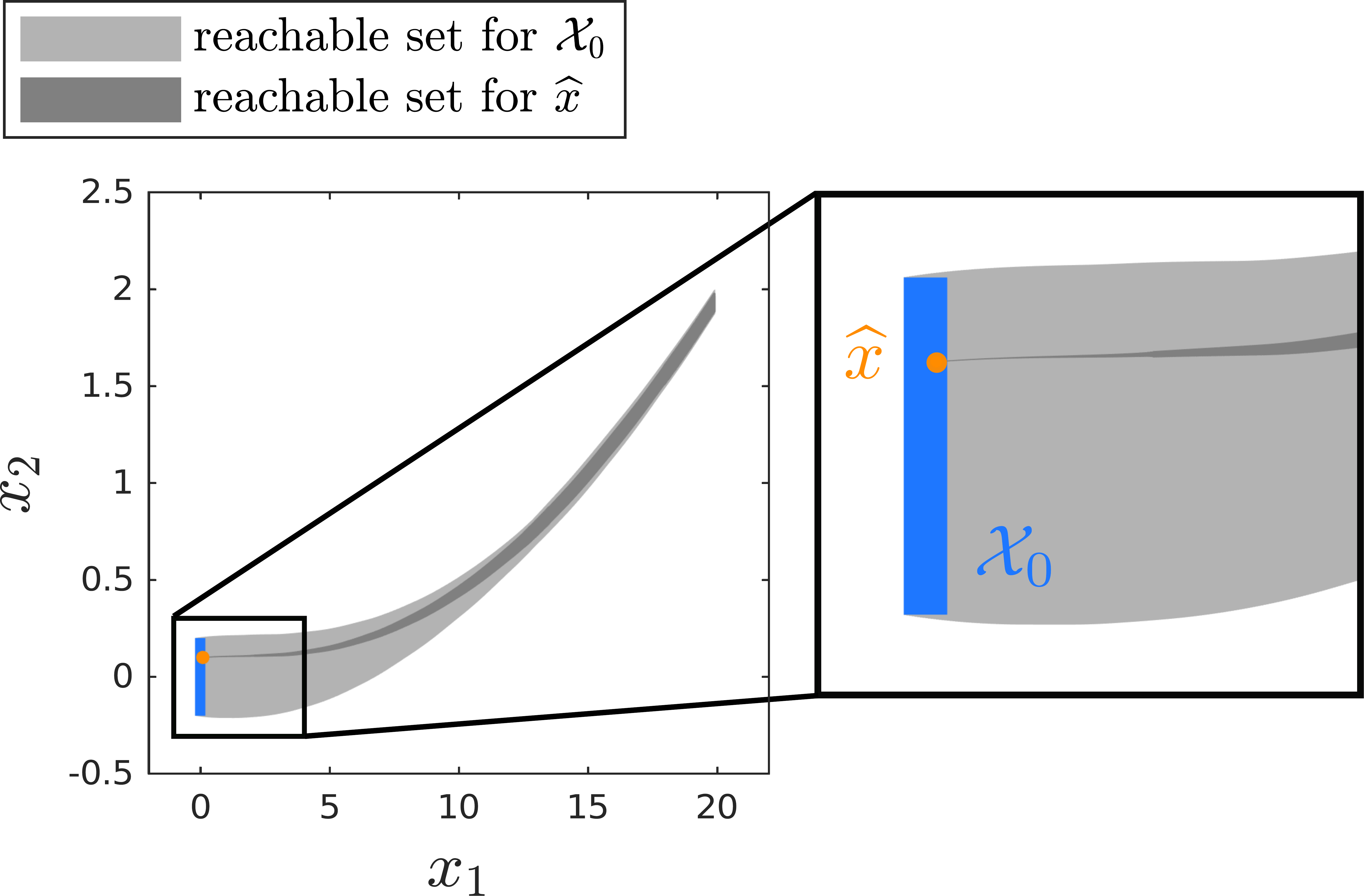}
	\caption{Visualization of Example~\ref{ex:control} with the initial set $\mathbf{\mathcal{X}_0}$ and the measured state $\widehat{x}$.}
	\label{fig:ExampleControl}
	\end{center}
\end{figure}

\begin{example}
	We consider the example of the turn-left maneuver of an autonomous car from {\normalfont \cite[Sec.~6]{Schuermann2017a}} with the measured velocity $\widehat{v} = 20.2 \frac{m}{s}$, the measured orientation $\widehat{\phi} = 0.01 rad$, and the measured position $\widehat{x}_1 = 0.1m$, $\widehat{x}_2 = 0.1m$. The reachable set for the initial set $\mathcal{X}_0$ and reachable set extracted for then measured point $\widehat{x} = [\widehat{v}~\widehat{\phi}~\widehat{x}_1~\widehat{x}_2]^T$ is visualized in Fig.~\ref{fig:ExampleControl}. As shown in Table~\ref{tab:compTime}, the extraction of the reachable set for the measured point from the offline-computed reachable set is significantly faster than the computation of the reachable set using Alg.~\ref{alg:reach}.
	\label{ex:control} 
\end{example}

\section{Conclusion}

In this paper, we showed that the computation of the reachable set for nonlinear systems with the conservative polynomialization approach using polynomial zonotopes preserves the dependence between the reachable states and the initial states. Since this novel method supports the efficient computation of reachable subsets inside pre-computed reachable sets, many possible applications are opening up. For the three applications \textit{safety falsification}, \textit{optimization over reachable sets}, and \textit{motion-primitive based control} we demonstrated with numerical examples that using our novel method results in significant speed-ups compared to the previous solutions. Besides, our method for extracting reachable subsets has complexity $\mathcal{O}(n^2)$ and is therefore computationally much more efficient than to compute the reachable subset from scratch, which has complexity $\mathcal{O}(n^5)$. Furthermore, this paper provides to some extent a method for unifying reachability analysis and falsification.

\begin{acks}
The author gratefully acknowledges financial support from the German Research Foundation (DFG) project faveAC under grant number AL 1185/5-1.
\end{acks}


\bibliographystyle{ACM-Reference-Format}
\bibliography{kochdumper,cpsGroup}
\end{document}